\documentclass[sigconf]{acmart}
 \pdfoutput=1
\usepackage{booktabs} 
\usepackage{graphicx}
\usepackage{caption}
\usepackage{subcaption}
\usepackage{amsthm}
\usepackage[linesnumbered,ruled,vlined,commentsnumbered]{algorithm2e}
\usepackage{latexsym}
\usepackage[]{algorithm2e}
\usepackage{adjustbox}
\usepackage{paralist}
\usepackage{amsmath}
\usepackage{color}

\newtheorem{mydef}{Definition}

\setcopyright{rightsretained}

\begin{document}
\title{On Rich Clubs of Path-Based Centralities in Networks}

\author{Soumya Sarkar}

 \affiliation{%
   \institution{Indian Institute of Technology}
   \city{Kharagpur} 
 }
 \email{soumya015@iitkgp.ac.in}

 \author{Sanjukta Bhowmick}
 \affiliation{%
   \institution{University of Nebraska}
   \city{Omaha} 
   }
 \email{sbhowmick@unomaha.edu}
 \author{Animesh Mukherjee}
 \affiliation{
   \institution{Indian Institute of Technology}
   \city{Kharagpur} 
   }
 \email{animeshm@cse.iitkgp.ernet.in}

\renewcommand{\shortauthors}{}

\begin{abstract}
Many scale-free networks exhibit a ``rich club'' structure, where high degree vertices form tightly interconnected subgraphs. In this paper, we explore the emergence of ``rich clubs'' in the context of shortest path based centrality metrics. We term these subgraphs of connected high closeness or high betweeness vertices as \textit{rich centrality clubs} (RCC).

Our experiments on real world and synthetic networks highlight the inter-relations between RCCs, expander graphs, and the core-periphery structure of the network. We show empirically and theoretically that  RCCs exist, if the core-periphery structure of the network is such that each shell is an expander graph, and their density decreases from inner to outer shells. 

We further demonstrate that in addition to being an interesting topological feature, the presence of RCCs is useful in several applications. The vertices in the subgraph forming the RCC are effective seed nodes for spreading information. Moreover, networks with RCCs are robust under perturbations to their structure.

Given these useful properties of RCCs, we present a network modification model that can efficiently create a RCC within networks where they are not present, while retaining other structural properties of the original network.  

The main contributions of our paper are:  {\bf (i)} we demonstrate that the formation of RCC is related to the core-periphery structure and particularly the expander like properties of each shell, {\bf(ii)} we show that the RCC property can be used to find effective seed nodes for spreading information and for improving the resilience of the network under perturbation  and, finally, {\bf(iii)} we present a modification algorithm that can insert RCC within networks, while not affecting their other structural properties. Taken together, these contributions present one of the first comprehensive studies of the properties and applications of rich clubs for path based centralities.

\end{abstract}

%
%



\maketitle

\section{Introduction}
In many social networks, the high degree nodes form a densely connected subgraph. This is known as the ``rich club'' phenomena. In this paper, we extend the definition of rich clubs, from high degree vertices, to shortest path based centralities, particularly high betweeness and high closeness centrality vertices. We term these extended rich clubs as {\em rich centrality clubs} (RCC). We present the global topological properties that lead to the formation of RCCs in complex networks (sections~\ref{sec:part2} and~\ref{sec:theory}). We further show how RCCs can be leveraged for spreading information efficiently and increasing network resilience (section~\ref{sec:application}). Finally we present a network modification algorithm to create RCCs in a network without disturbing other structural properties of the original network(section~\ref{sec:part3}).   

Our study is motivated by the fact that over the last few years several papers\cite{shin2016corescope,Shin2017,li2015correlation,kitsak2010identification} have independently reported that vertices in the inner shells of the networks can be leveraged to identify high centrality nodes or serve as seeds for community detection. However, each paper focused on only one type of analysis and there was rarely any overlap between the networks studied in these papers. When we conducted an integrated study over a large set of real-world and synthetic networks, we observed that the reported properties of the  vertices in the inner shells hold only for a  certain type of networks. This observation impelled us to investigate the topological property of networks where the inner shells  contain high centrality nodes.

We observed that the inner shells of networks are typically dense, thus if they contain high centrality nodes, then by virtue of being dense, these cores would form a rich centrality club. However, unlike degree which is a local variable, closeness 
and betweeness centralities are based on shortest paths which are global variables. Building on this observation we demonstrate that the networks with RCC also maintain a global pattern. Specifically, each shell is an expander graph, and going from the inner to the outer shell,the shells have gradually decreasing density. In other words, visually and quantitatively, networks with RCC  expand out from a dense inner core to sparse outer shells (see Figure~\ref{motivationfig_1}).



The presence of rich centrality clubs confers several favorable properties to the networks. In particular, due to the presence of many high path based centrality vertices within a small subgraph, the vertices in the RCC can be effective seed nodes in quickly spreading information across the network. Moreover, similar to the traditional rich club, the presence of RCC increases the resilience of the networks under edge perturbations.

Given these favorable properties, we posit that, in many cases, the presence of RCC is desirable. To this end, we propose  
 a modification model that can form a RCC in a network where it is absent. Our model is such that other properties of the original network including the power law exponent, the average degree, shortest path based centralities and clustering co-efficient remain unchanged. 

The  {\bf key contributions} of our paper are as follows.
\begin{compactitem}
\item [(i)] We study the formation of {\em rich  clubs of shortest path based centralities} in complex networks and observe that their presence can lead to faster identification of high centrality nodes and communities. We  demonstrate theoretically  and empirically  that {\em in networks containing RCC, the  shells are expander-like and the density of the shells decreases from the inner to the outer shells} (sections~\ref{sec:part2} and~\ref{sec:theory}).


\item [(ii)] We empirically show that networks containing RCC have several favorable properties (section~\ref{sec:application}). Specifically, {\em the vertices within the RCC are effective seed nodes} for information spreading and {\em networks  containing RCC are resilient} to perturbations to their edges. 

\item [(iii)] We propose a {\em modification model that can insert RCC} into a network, while maintaining other structural properties of the original network (section~\ref{sec:part3}).  Our model is reversible in that when operations are applied in reverse (deletion instead of addition of edges), the RCC can be removed from a network, while also maintaining the other structural properties. Our model only requires the information of the degree of the vertices, which is a much faster operation than computing the betweenness and closeness centralities.

\end{compactitem}

 \begin{figure*}
  
  \begin{minipage}[b]{\textwidth}
  	\centering
    \includegraphics[width=0.35\linewidth, height=0.18\textheight]{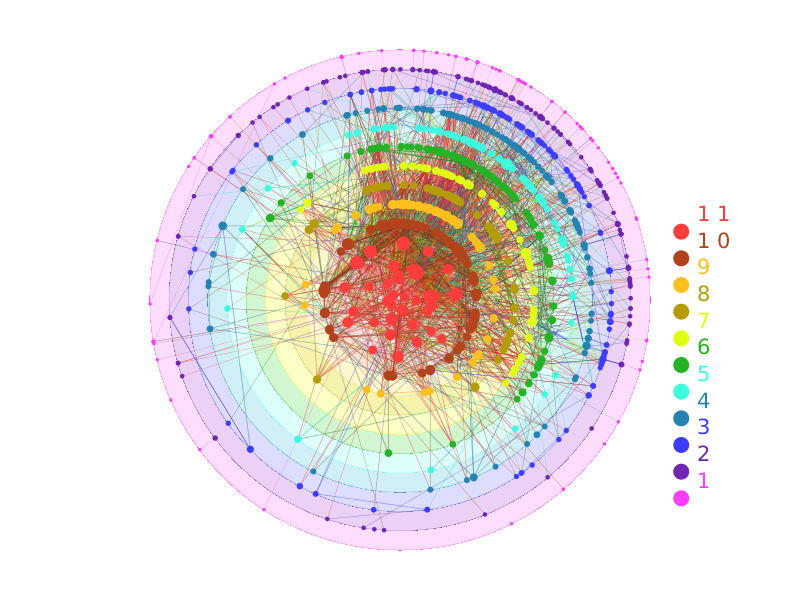}
     \includegraphics[width=0.35\linewidth, height=0.18\textheight]{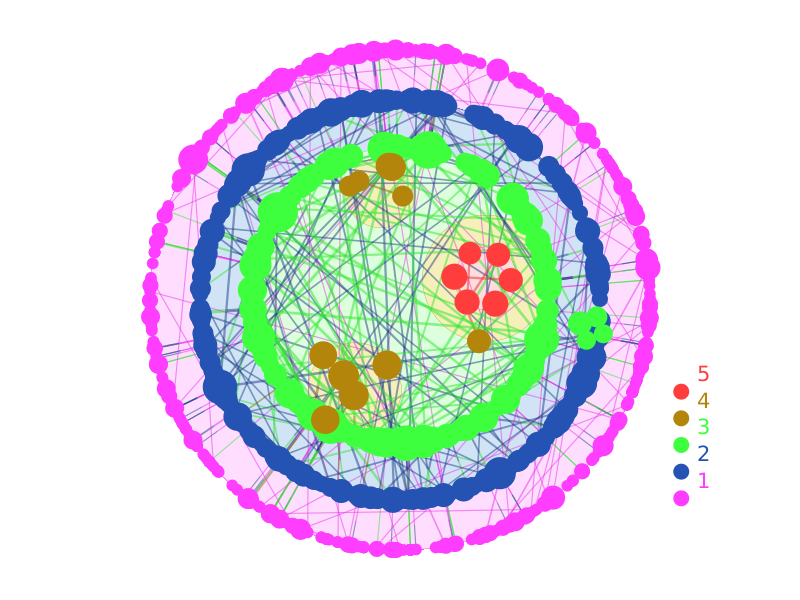}
   \end{minipage}
   \caption{\label{motivationfig_1}Comparing the core periphery of networks generated by~\cite{alvarez2006lanet}. Left: Network of software dependencies demonstrates presence of RCC; shells are arranged as concentric cycles. Right: Network of protein interactions in yeast does not demonstrate the presence of RCC; the innermost cores are not at the center of the network.Color online}
\end{figure*}

\section{Definitions and Datasets}
\label{sec:prelims}
We briefly describe the definitions of the network properties and the test suite used for the experiments.

\subsection{Definitions of network properties}

\begin{mydef}
{\bf $k$-core}: Given a graph $G(V,E)$, where $V,E$ are the set of vertices and edges, a k-core is a maximal subgraph $G_{k-core}$ such that each node in $G_{k-core}$
has degree at least k.
\end{mydef}
\begin{mydef}
{\bf $k$-shell}: Given a graph $G(V,E)$, a $k$-shell is the
induced subgraph over the maximal set of nodes such that (1)
the $k$-shell does not include nodes from any existing higher
shells, and (2) each node in the
$k$-shell has at least $k$ connections to nodes in the $k$-shell or
the higher shells. 
\end{mydef}

The core number of a node is the highest value $k$ such
that the node is a part of a $k$-core. The $k$-core decomposition~\cite{govindan2017k} is the assignment of core numbers to nodes. Core numbers can be computed with complexity of $O(|E|)$.

To find a $k$-core, the $k$-core decomposition algorithm recursively removes nodes with degree less than $k$. We will denote all the nodes belonging to the $k$-core by the set $C_k$. Note that if $G$ is connected then $C_1$ is equivalent to the $V$. Subsequent inner cores are part of the outer cores i.e $ C_k \subset .....C_3 \subset C_2 \subset C_1 \subset V $

Note that the $k$-core is the induced subgraph of the union
of $j$-shells for $j \geq k$. Hence the shells of the graph partitions the set $V$ into disjoint sets of vertices. 



\begin{mydef}
{\bf Expansion property}~\cite{malliaros2011expansion}: Given a graph $G = (V, E)$ where $V$ is the set of nodes and $E$ is the set of edges, the expansion of a set of nodes $S \subset V$ is a function of the number of nodes in $V-S$ to which $S$ is connected. That is, if $N(S)$ is the set of nodes to which $S$ is connected, then the expansion of $S$ is $\frac{|N(S)|}{|S|}$ 
\end{mydef} 
A graph with a high value of expansion property is known as an {\em expander graph}. In an expander graph any subset $S \in V $ (where $ |S| \le \frac{V}{2} $) will have many neighbors. The expansion of a graph can be measured using the {\em Cheeger constant}, $h(G)$. A low Cheeger constant indicates a "bottle neck", i.e. the graph can be partitioned by removing very few edges. A high Cheeger constant indicates that no matter how the graph is partitioned, the number of edges across the partitions is always large.

Accurate calculation of the Cheeger constant is NP hard. For d-regular graphs it can be approximated by the second smallest eigenvalue of the spectrum of the normalised graph Laplacian given by $L = I - D^{-\frac{1}{2}}AD^{-\frac{1}{2}}$. The second smallest eigenvalue $(\lambda_{2})$, also known as the eigengap, is related to Cheeger constant by $\frac{\lambda_{2}}{2} \leq h(G) \leq \sqrt[]{2\lambda_2}$  which is also know as Cheeger's inequality. It was shown in~\cite{Chung} that the lower bound also holds for general graphs.

\subsection{Test suite of networks}
We used a diverse set of networks for our experiments. We used 8 real world networks that are publicly available~\cite{snapnets,kunegis2013konect}.  We generated 15 synthetic networks of varying sizes and varying network core structures using the MUSKETEER synthetic network generation tool ~\cite{gutfraind2015multiscale}. A summary of the properties of these 23 networks is given in Table~\ref{dataset}. The $\alpha$ is the scale free exponent obtained after fitting the empirical degree distribution to a power law distribution which is given by $p(k) \sim k^{-\alpha}$  where $p(k)$ is the fraction of nodes having degree $k$. All the networks are considered to be undirected.

\begin{table}\footnotesize
\begin{tabular}{ |c|c|c|c|c|c|c|c| } 
 \hline
 Network & Nodes & Edges & $\alpha$ & $\mu(d_v)$ & $\mu(C_lC)$ & $\mu(BC)$ & \textit{LCN}\\ 
 \hline
 AS~\cite{snapnets} & 6474 & 13895 & 1.235 & 4.29 & 0.27 & 0.004 & 12\\ 
 \hline
 Caida~\cite{snapnets} & 16493 & 33372 & 1.17 & 4.04 & 0.27 & 0.001 & 20\\ 
 \hline
 Bible~\cite{kunegis2013konect} & 1707 & 9059 & 1.523 & 10.61 & 0.31 & 0.001 & 15\\ 
 \hline
 Software~\cite{kunegis2013konect} & 994 & 4645 & 1.168 & 9.32 & 0.34 & 0.002 & 11 \\ 
 \hline
 Protein~\cite{kunegis2013konect} & 1458 & 1993 & 2.106 & 2.73 & 0.15 & 0.004  & 5 \\ 
 \hline
 Facebook~\cite{snapnets} & 7178 & 10298 & 2.896 & 2.86 & 0.11 & 0.001 & 5 \\ 
 \hline
 Hepth~\cite{snapnets} & 2694 & 4255 & 1.487 & 3.15 & 0.18 & 0.001 & 7\\ 
 \hline
 Power~\cite{kunegis2013konect} & 4941 & 6594 & 2.845 & 2.66  &0.05 & 0.003 & 5 \\ 
 \hline
 N1 & 14212 & 34901 & 1.215 & 16.3 & 0.25  & 0.0007 & 16 \\ 
 \hline
 N2 & 10162 & 25154 & 1.28 & 4.95  &0.27 & 0.0008 & 14 \\ 
 \hline
 N3 & 36469 & 96990 & 1.237 & 2.66  & 0.24 & 0.003 & 27 \\ 
 \hline 
 N4 & 65630 & 170061 & 2.29 & 2.66  &0.13 & 0.003 & 12 \\ 
 \hline 
 N5 & 4091 & 33352 & 1.438 & 2.66  &0.33 & 0.003 & 21 \\ 
 \hline 
 N6 & 6785 & 44381 & 1.421 & 2.66  &0.31 & 0.003 & 19 \\ 
 \hline 
 N7 & 6009 & 13585 & 2.016 & 2.66  &0.18 & 0.003 & 6 \\ 
 \hline 
 N8 & 3863 & 22356 & 1.268 & 2.66  &0.34 & 0.003 & 23 \\ 
 \hline 
 N9 & 11278 & 19616 & 2.737 & 3.47  &0.03 & 0.003 & 5 \\ 
 \hline 
 N10 & 19623 & 33711 & 2.832 & 3.43  &0.05 & 0.00049 & 5 \\ 
 \hline 
 N11 & 5980 & 9501 & 3.027 & 3.17  &0.05 & 0.05 & 6 \\ 
 \hline 
 N12 & 6045 & 13592 & 2.373 & 4.49  &0.11 & 0.0035 & 7 \\ 
 \hline 
 N13 & 7783 & 35185 & 2.517 & 3.61  &0.06 & 0.0033 & 6 \\ 
 \hline 
 N14 & 15988 & 28373 & 3.029 & 3.54  &0.09 & 0.004 & 5 \\
 \hline 
 N15 & 28651 & 51159 & 3.073 & 3.57  & 0.04 & 0.0028 & 6 \\ 
 \hline
\end{tabular}
\caption {\label{dataset} Test suite of networks and their properties. $\alpha$: power-law exponent, $\mu(d_v)$: average degree, $\mu(C_lC)$: average clustering co-efficient, $\mu(BC)$: average betweenness centrality. (LCN):  largest core number in the network.}
\end{table}


\section{Motivating Experiments}\label{sec:part1}
We present the experiments that motivated our research. We test whether  vertices with high core numbers {\bf(i)} have high centralities and {\bf(ii)} can be used as seed nodes for community detection.


\subsection{Correlation with other centrality metrics}\label{sec:correl}

Several papers~\cite{wildie12,holme05,silva08,lin14,Meyer2015}, claim that the vertices with high core numbers should also have high centrality values. To test this claim, we compute the Jaccard coefficient ($J_c$) given by $\frac{S_1 \cap S_2}{ S_1 \cup S_2 }$ between the set of vertices with highest core numbers ($S_1$) and an 
equal number of high ranked nodes for each of the centrality metrics ($S_2$).



The results in Table~\ref{tab:jacs} show a clear separation of the networks. In the first group, all the networks (\textcolor{blue}{blue}) have high $J_c$ implying significant number of high central nodes also have highest core numbers. In the second group, (\textcolor{brown}{brown}) the high core numbered nodes do not have high centrality as per the low $J_c$ scores. 



\begin{table}
\begin{tabular}{ |c|c|c|c| } 
 \hline
 Network & degree & closeness & betweenness \\ 
 \hline
 \textcolor{blue}{AS} & 0.6 & 0.75 & 0.5 \\ 
 \hline
 \textcolor{blue}{Caida} & 0.56 & 0.69 & 0.49 \\ 
 \hline
 \textcolor{blue}{Bible} & 0.49 & 0.64 & 0.43 \\ 
 \hline
 \textcolor{blue}{Software} & 0.40 & 0.5 & 0.23 \\ 
 \hline
 \hline
 \textcolor{brown}{Protein} & 0 & 0 & 0 \\ 
 \hline
 \textcolor{brown}{Facebook} & 0.10 & 0 & 0 \\ 
 \hline
 \textcolor{brown}{Power} & 0 & 0 & 0 \\ 
 \hline
 \textcolor{brown}{Hepth} & 0.05 & 0.05 & 0.05 \\ 
 \hline
 \hline
 \textcolor{blue}{N1} & 0.63 & 0.86 & 0.68 \\ 
 \hline
 \textcolor{blue}{N2} & 0.722 & 0.530 & 0.653 \\ 
 \hline
 \textcolor{blue}{N3} & 0.74 & 0.78 & 0.80 \\ 
 \hline
 \textcolor{blue}{N4} & 0.68 & 0.81 & 0.68 \\
 \hline
 \textcolor{blue}{N5} & 0.80 & 0.82 & 0.78 \\ 
 \hline
 \textcolor{blue}{N6} & 0.57 & 0.71 & 0.6 \\ 
 \hline
\textcolor{blue}{N7} & 0.39 & 0.48 & 0.37 \\ 
 \hline
 \textcolor{blue}{N8} & 0.79 & 0.79 & 0.79 \\ 
 \hline
 \hline
 \textcolor{brown}{N9} & 0 & 0 & 0 \\ 
 \hline
 \textcolor{brown}{N10} & 0 & 0 & 0 \\ 
 \hline
 \textcolor{brown}{N12} & 0.02 & 0 & 0 \\ 
 \hline
 \textcolor{brown}{N13} & 0.06 & 0 & 0 \\ 
 \hline
 \textcolor{brown}{N14} & 0 & 0 & 0 \\ 
 \hline
 \textcolor{brown}{N15} & 0.002 & 0 & 0 \\ 
 \hline
\end{tabular}
\caption {\label{tab:jacs} Jaccard index between nodes with highest coreness and equal number of high centrality nodes. Results clearly separate the two categories of networks into ones that have an RCC (\textcolor{blue}{blue}) and ones that do not have an RCC (\textcolor{brown}{brown}).}
\end{table}

\subsection{High core numbers to detect communities}
\label{sec:highc}
In existing techniques of community detection utilizing $k$-core structure~\cite{peng2014accelerating},
the network is reduced to its $k$-core subgraph, for a predetermined value of $k$, and the communities in the subgraph are computed. The vertices in  these communities are  used as seed nodes to propagate the community information to other vertices. 

We note that this process will succeed only if the communities are well represented in the reduced network. Therefore to test the applicability of this algorithm, we tabulate how the vertices in the innermost core are distributed across the communities. 


 Figure~\ref{fig:comms} plots the community ids of the networks in the x-axis and the number of nodes from the innermost core that are members of a particular community in the y-axis. 
 We include all communities whose at least one vertex is in the innermost core.
 

We again observe a separation between two class of networks. In one group, the vertices from the innermost core are spread over multiple communities, whereas in the other group, the vertices from the innermost core are concentrated in one or two, communities. Clearly, the first group is more suitable for the community detection algorithm described earlier. Figure~\ref{fig:commvis} shows the community distribution in the innermost core of two networks from our test suite. These results demonstrate that {vertices with high core numbers are not always distributed across multiple communities, and in some cases can be concentrated in only one community.}


\begin{figure*}[t]
  \begin{subfigure}[b]{0.15\textwidth}
    \includegraphics[width=\textwidth]{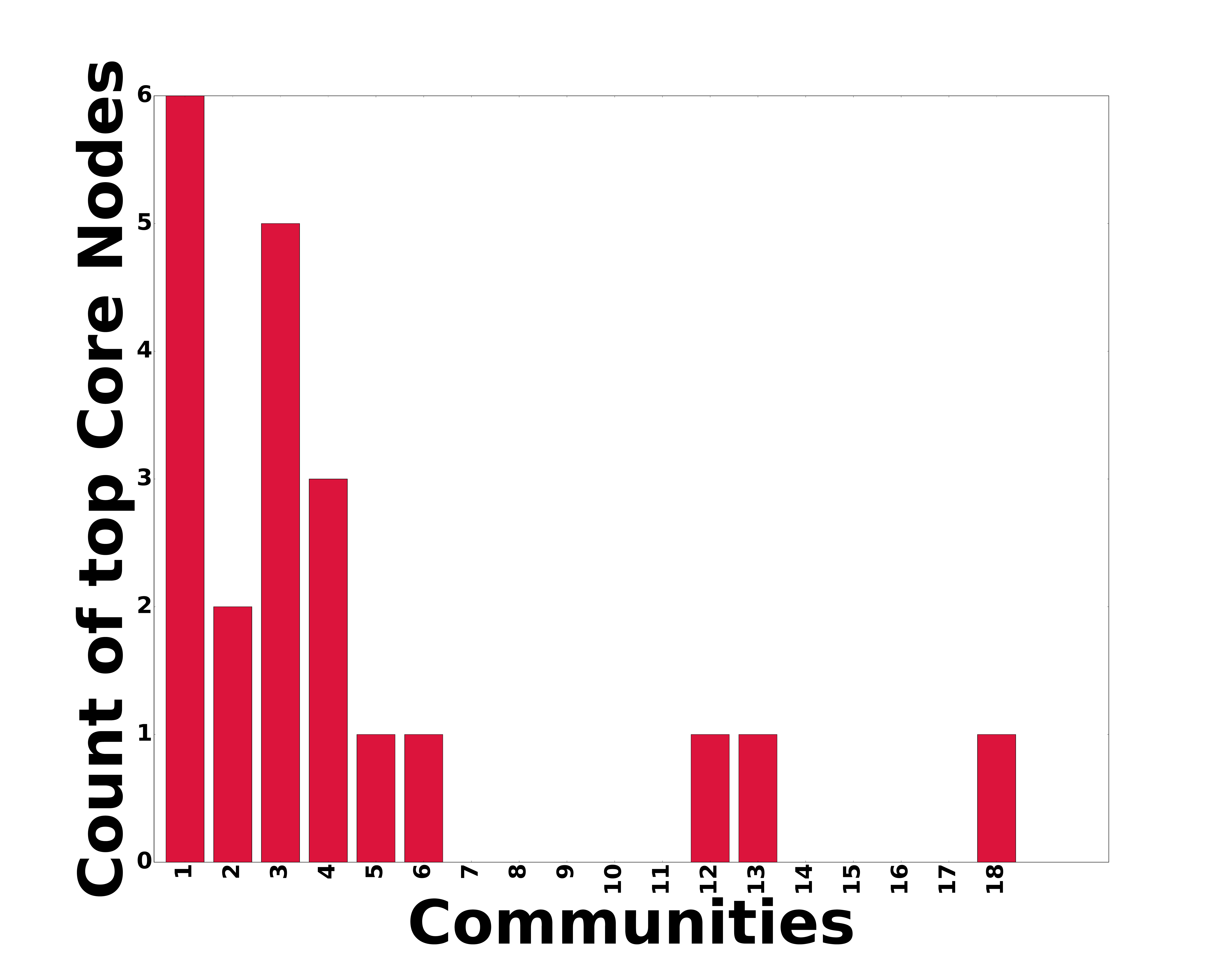}
    \caption{AS}
  \end{subfigure}
  \begin{subfigure}[b]{0.15\textwidth}
    \includegraphics[width=\textwidth]{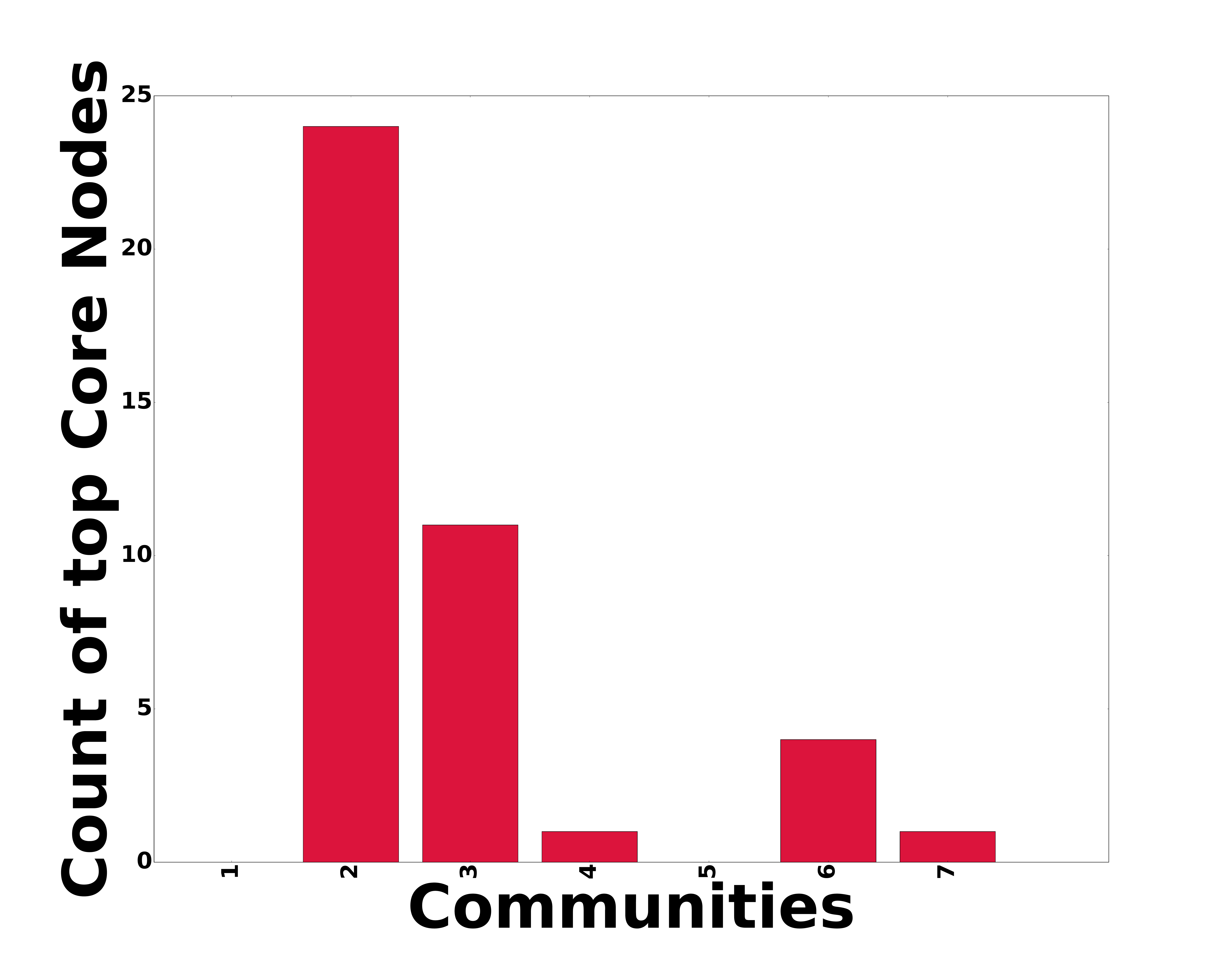}
    \caption{Bible}
  \end{subfigure}
  \begin{subfigure}[b]{0.15\textwidth}
    \includegraphics[width=\textwidth]{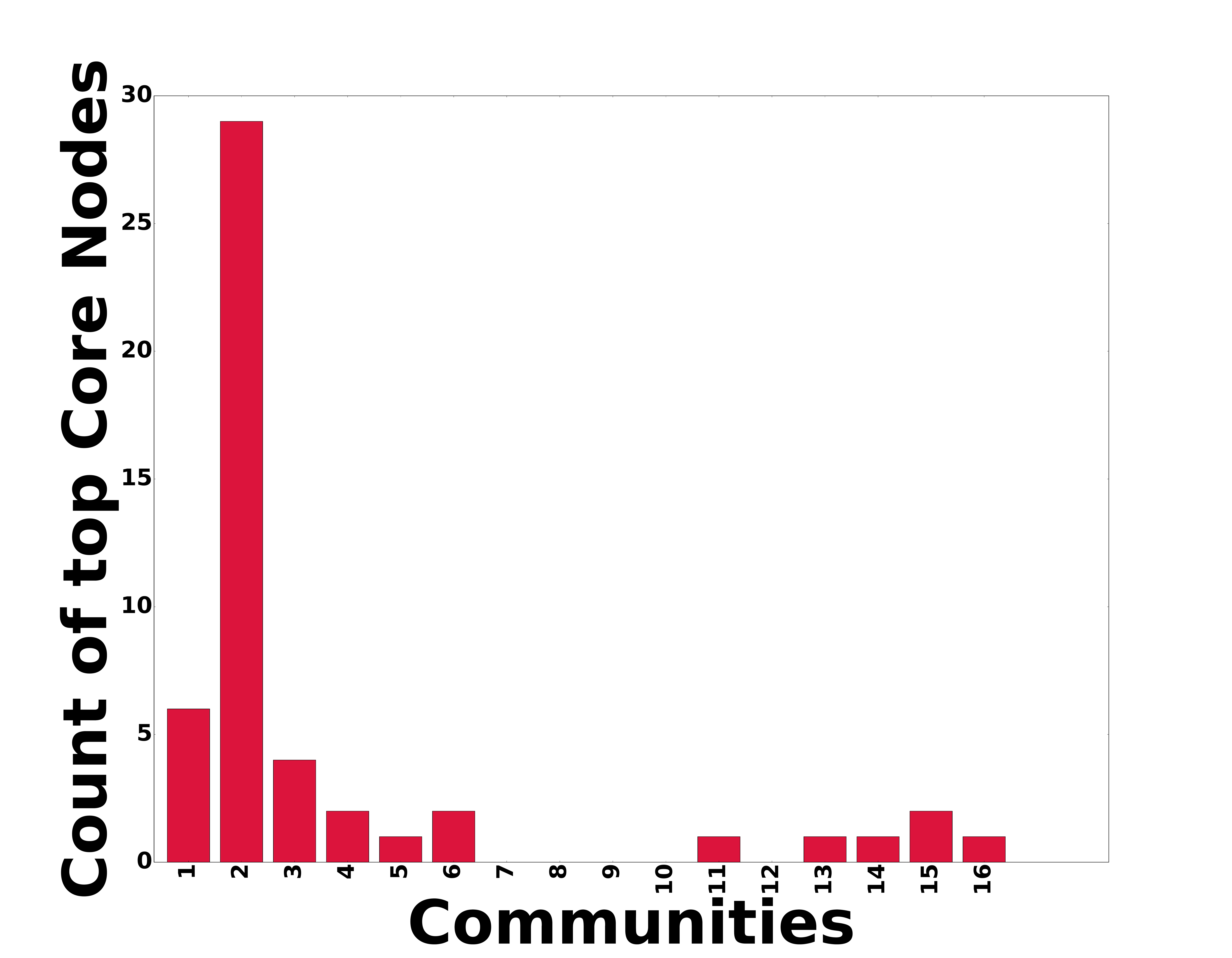}
    \caption{Caida}
  \end{subfigure}
  \begin{subfigure}[b]{0.15\textwidth}
    \includegraphics[width=\textwidth]{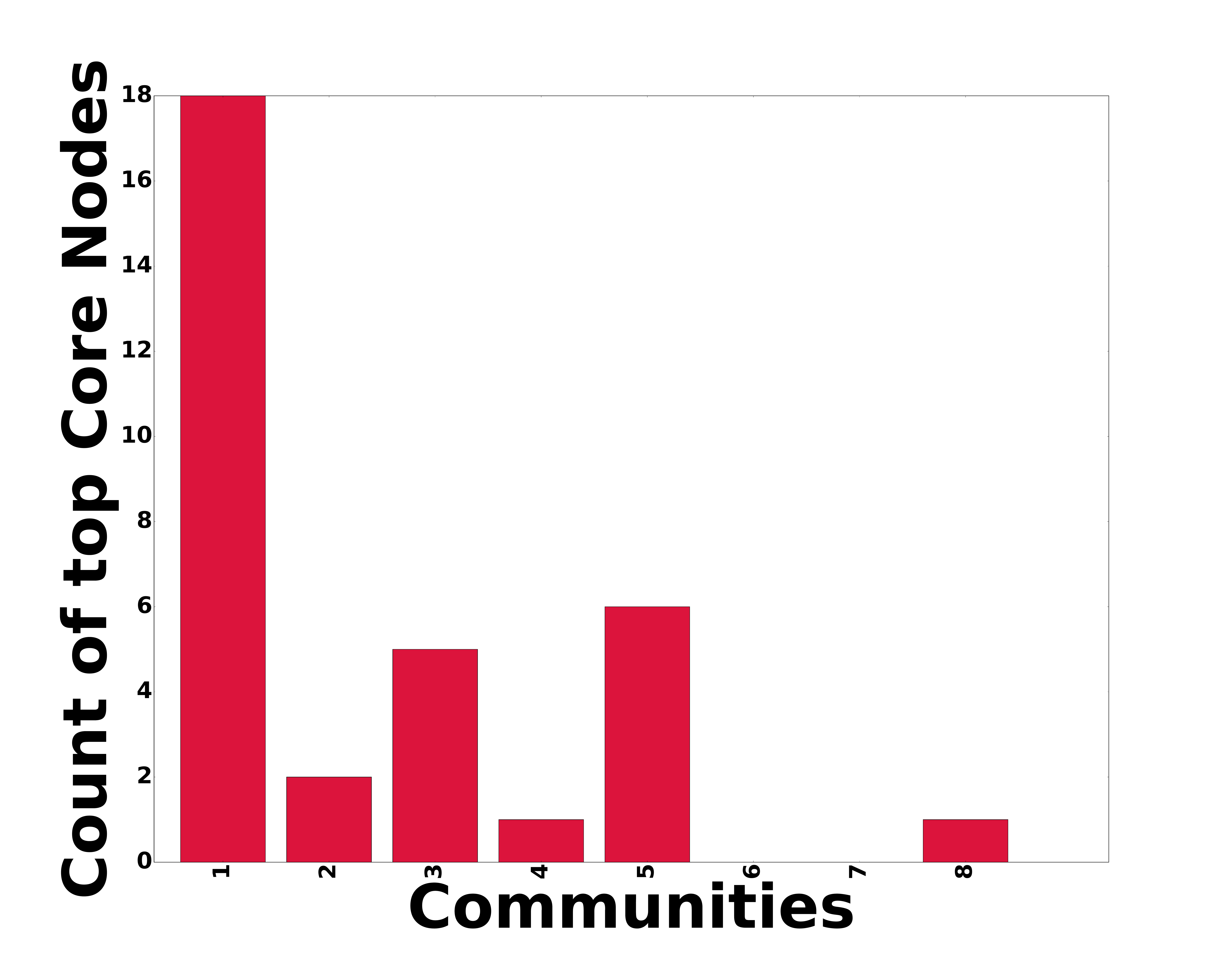}
    \caption{Software}
  \end{subfigure}
  \begin{subfigure}[b]{0.15\textwidth}
    \includegraphics[width=\textwidth]{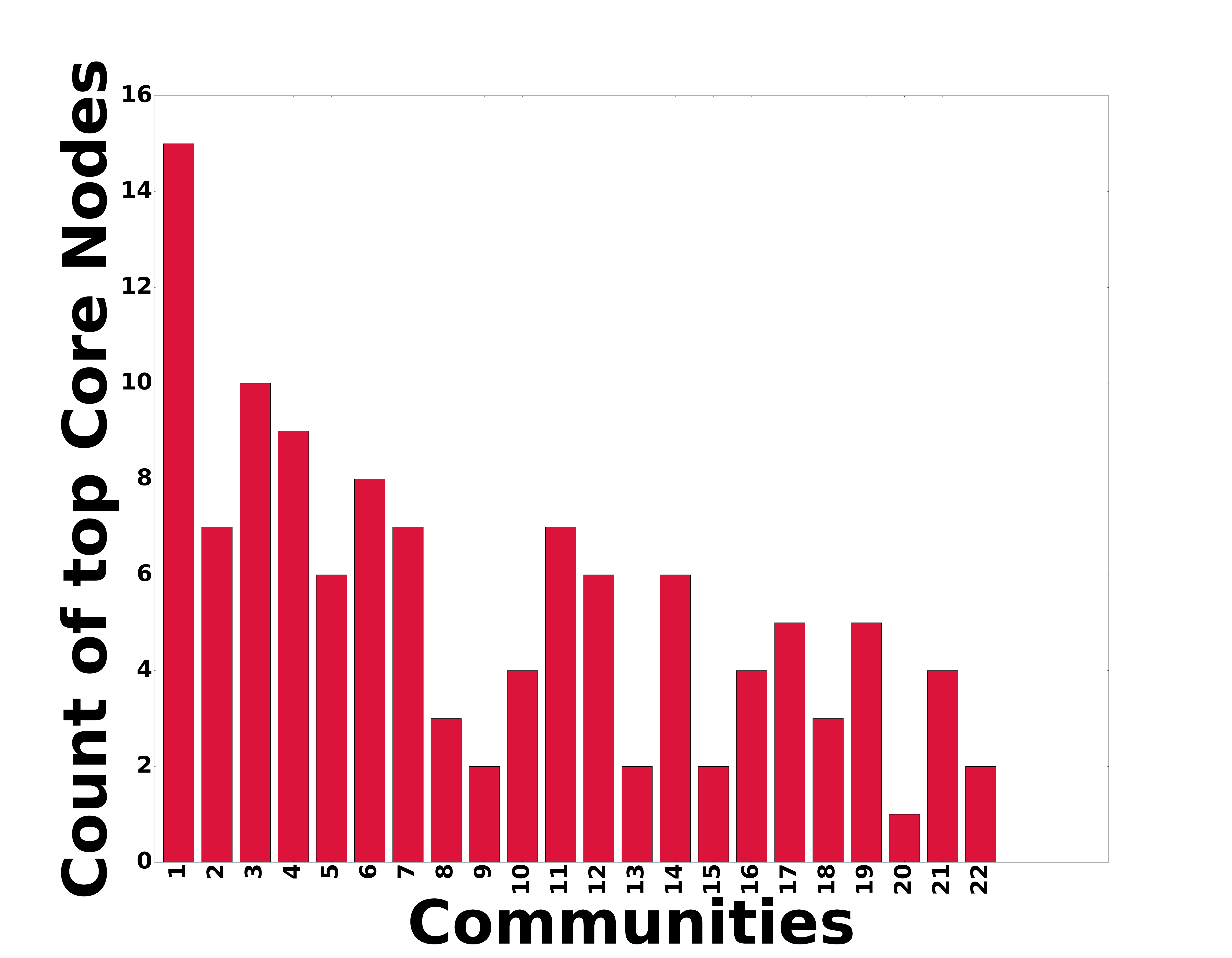}
    \caption{N6}
  \end{subfigure}
  \begin{subfigure}[b]{0.15\textwidth}
    \includegraphics[width=\textwidth]{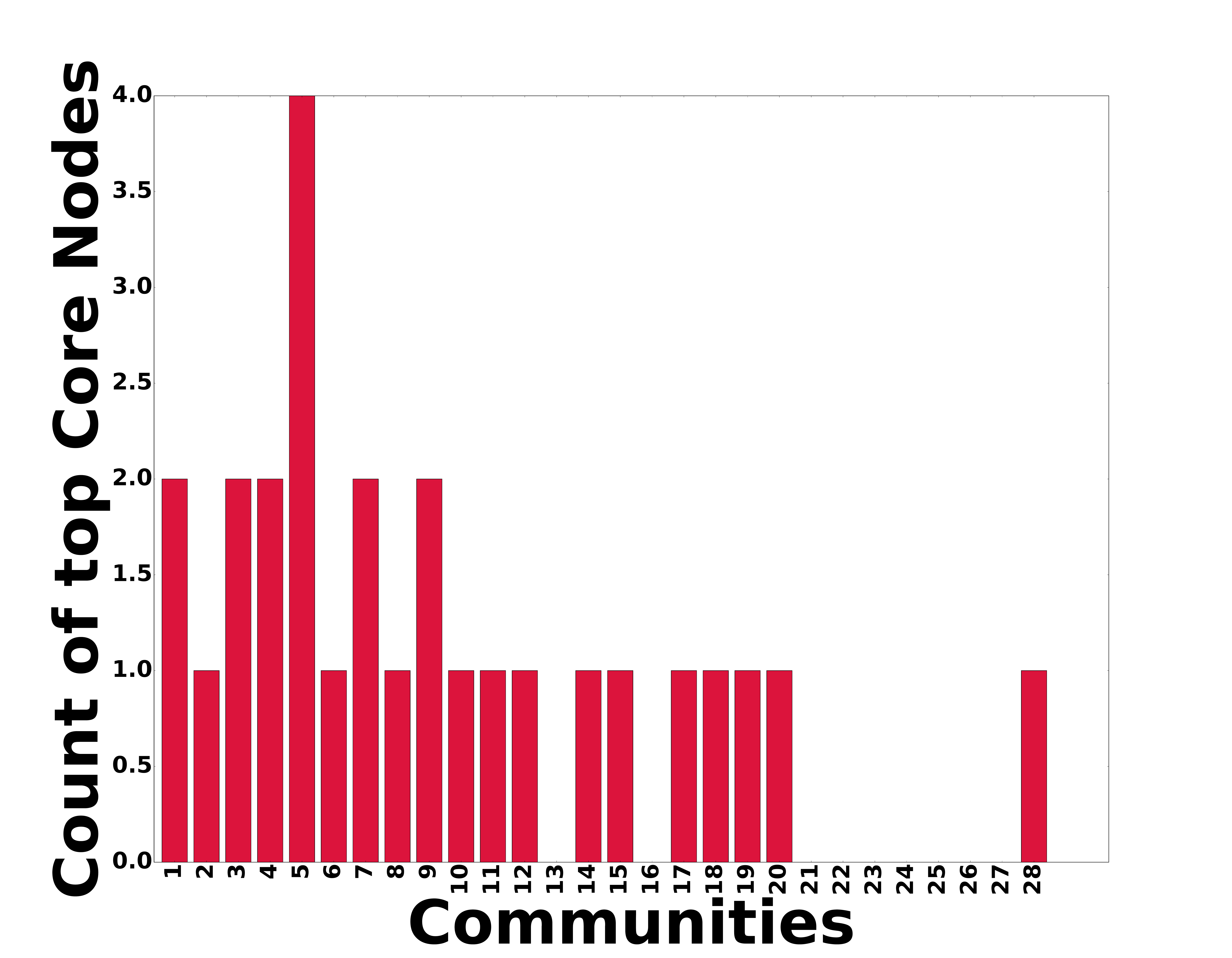}
    \caption{N7}
  \end{subfigure}

  \begin{subfigure}[b]{0.15\textwidth}
    \includegraphics[width=\textwidth]{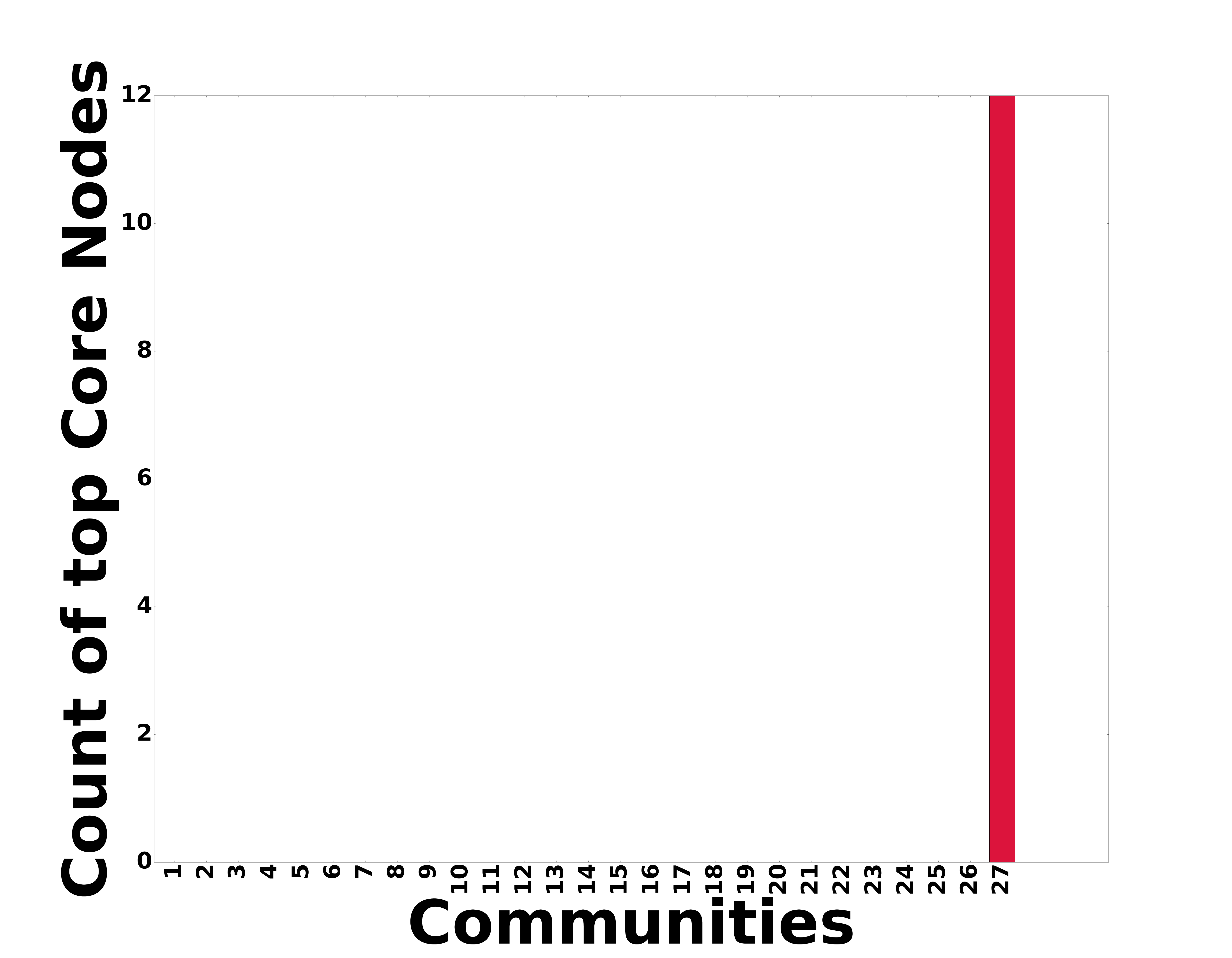}
    \caption{Power}
  \end{subfigure}
  \begin{subfigure}[b]{0.15\textwidth}
    \includegraphics[width=\textwidth]{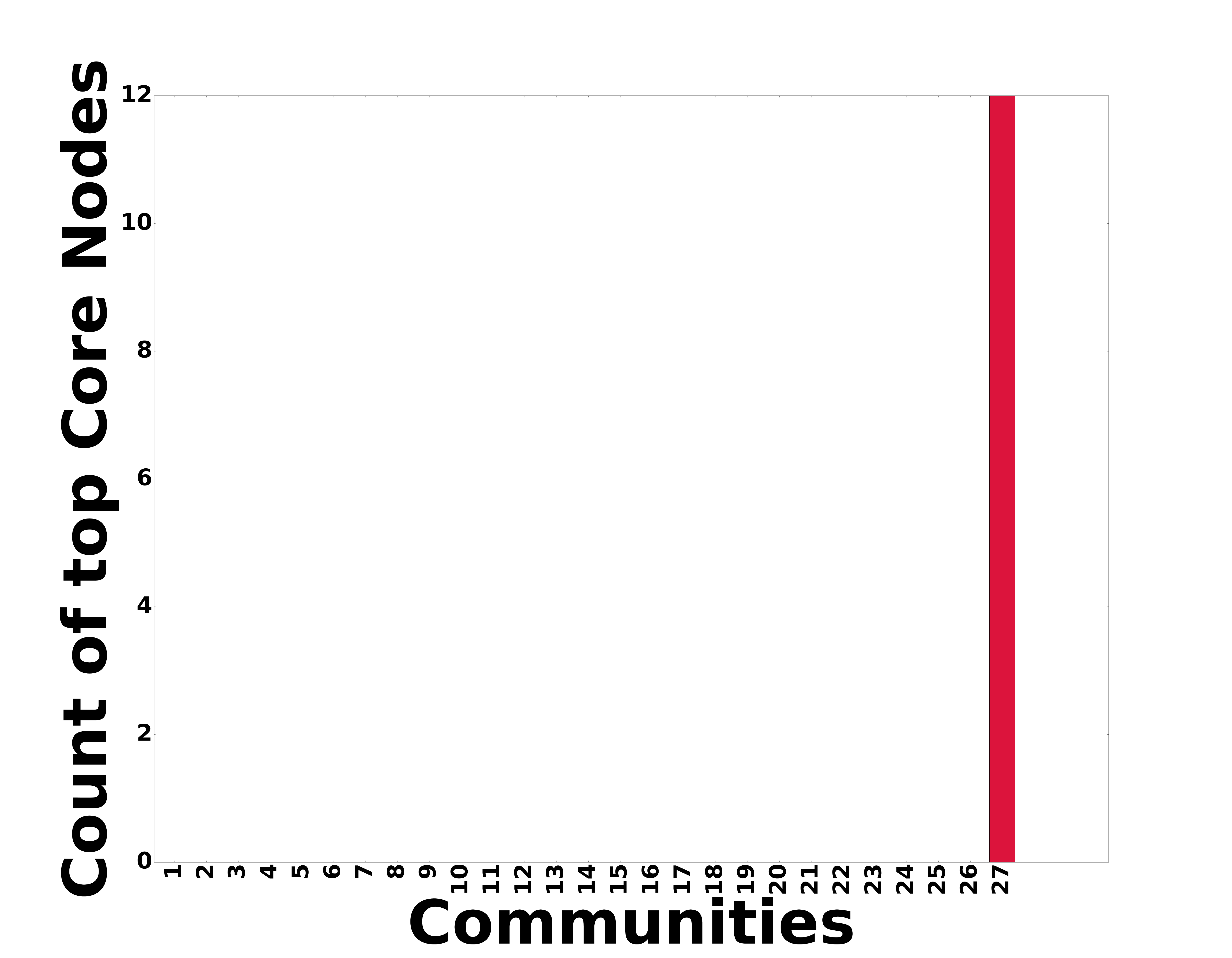}
    \caption{Protein}
  \end{subfigure}
  \begin{subfigure}[b]{0.15\textwidth}
    \includegraphics[width=\textwidth]{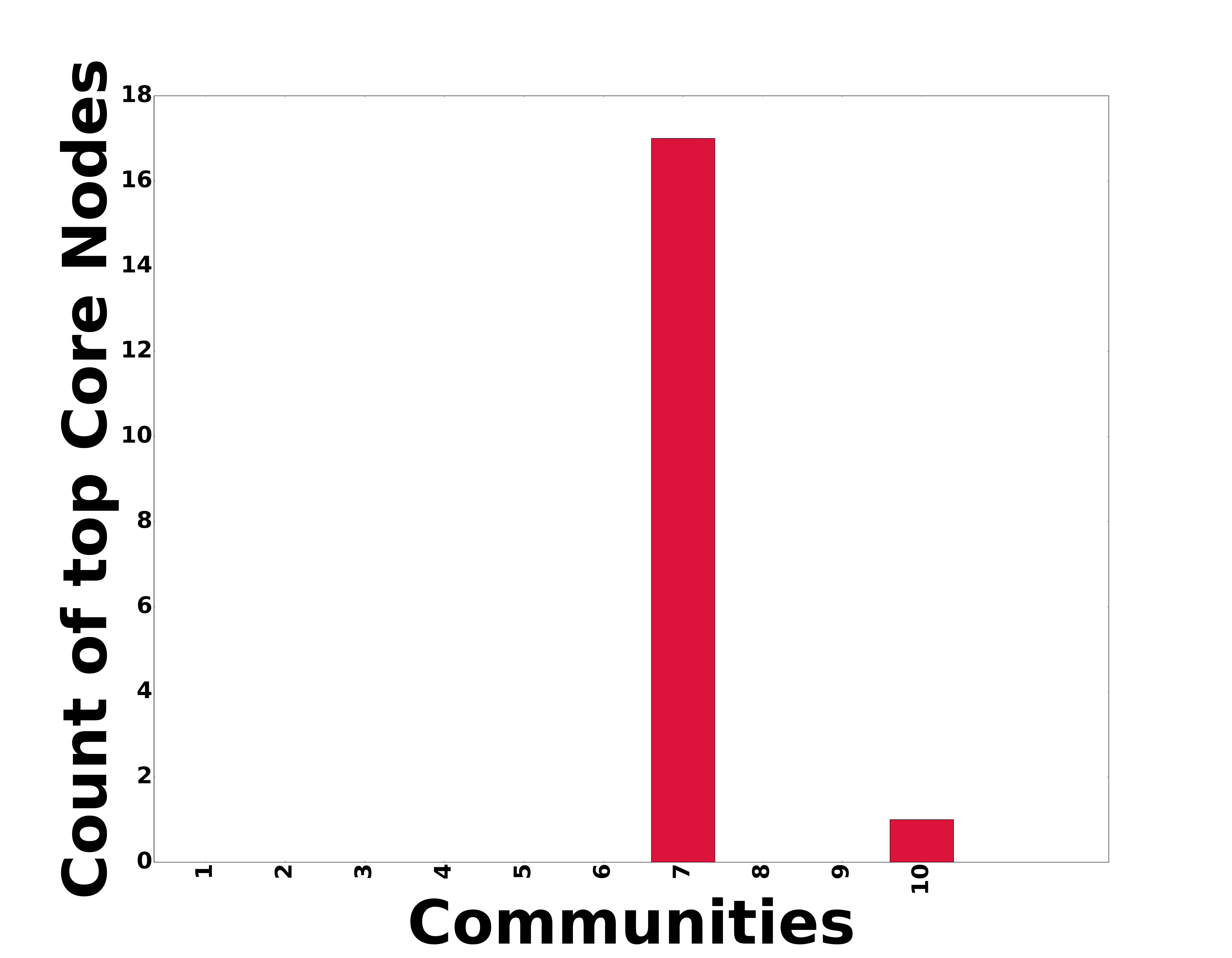}
    \caption{Hepth}
  \end{subfigure}
  \begin{subfigure}[b]{0.15\textwidth}
    \includegraphics[width=\textwidth]{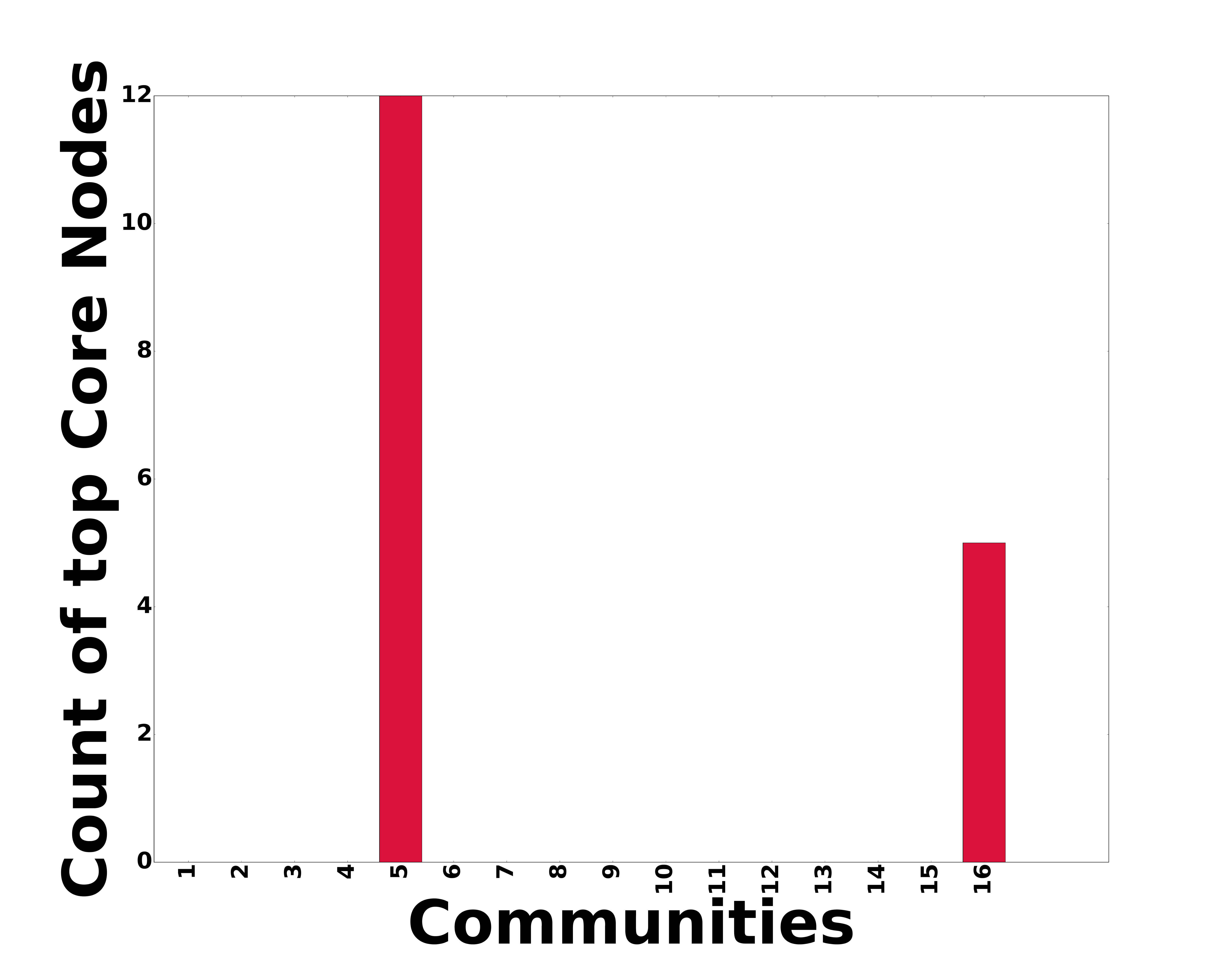}
    \caption{Facebook}
  \end{subfigure}
  \begin{subfigure}[b]{0.15\textwidth}
    \includegraphics[width=\textwidth]{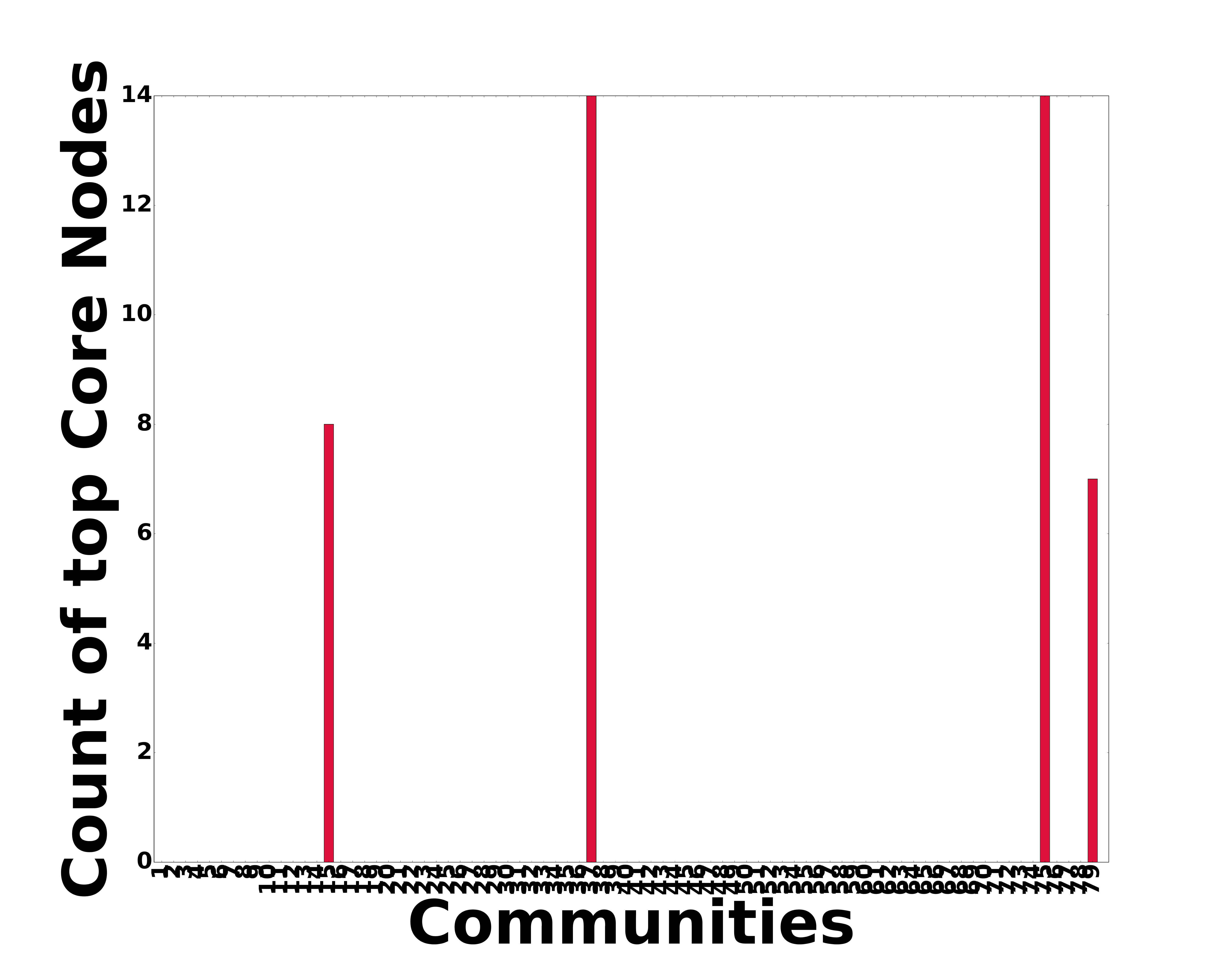}
    \caption{N8}
  \end{subfigure}
  \begin{subfigure}[b]{0.15\textwidth}
    \includegraphics[width=\textwidth]{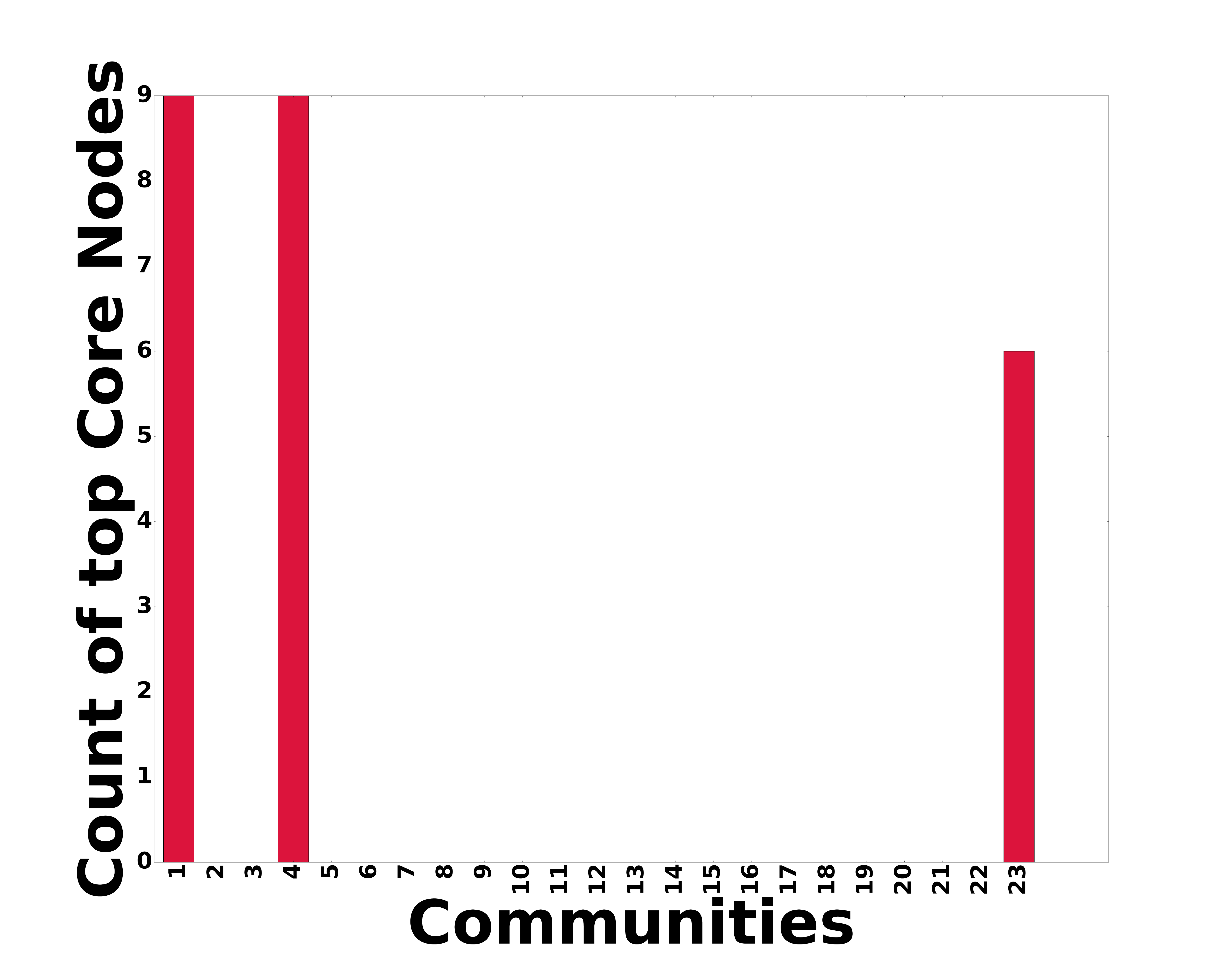}
    \caption{N9}
  \end{subfigure}
  \caption{\label{fig:comms}Distribution of innermost core nodes in the different communities of the network. The X-axis indicates the community ids of a network ordered in terms of number of nodes present in that community. Y-axis indicates the number of innermost core nodes in a particular community with the id on the x-axis. The X-axis stretches includes all communities that contain the nodes from the  innermost core. Note that for networks with an RCC (top) the high core vertices are more spread out among multiple communities as compared to those that do not have an RCC (bottom).} 
\end{figure*}

\begin{figure}[t!]
    \centering
    \begin{subfigure}[t]{0.25\textwidth}
        \centering
        \includegraphics[height=2in]{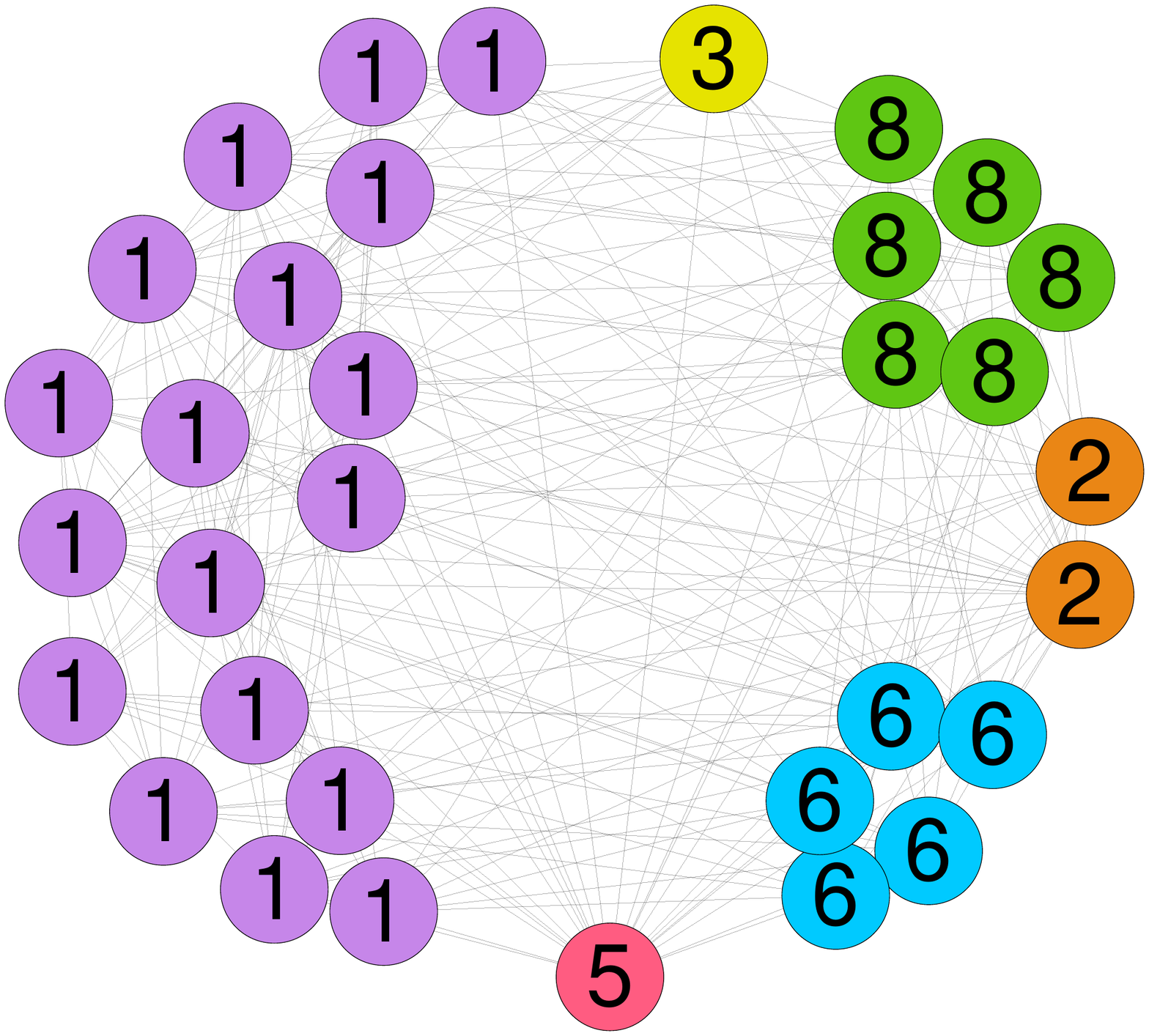}
        \caption{Software}
    \end{subfigure}%
    ~
    \begin{subfigure}[t]{0.25\textwidth}
        \centering
        \includegraphics[height=2in]{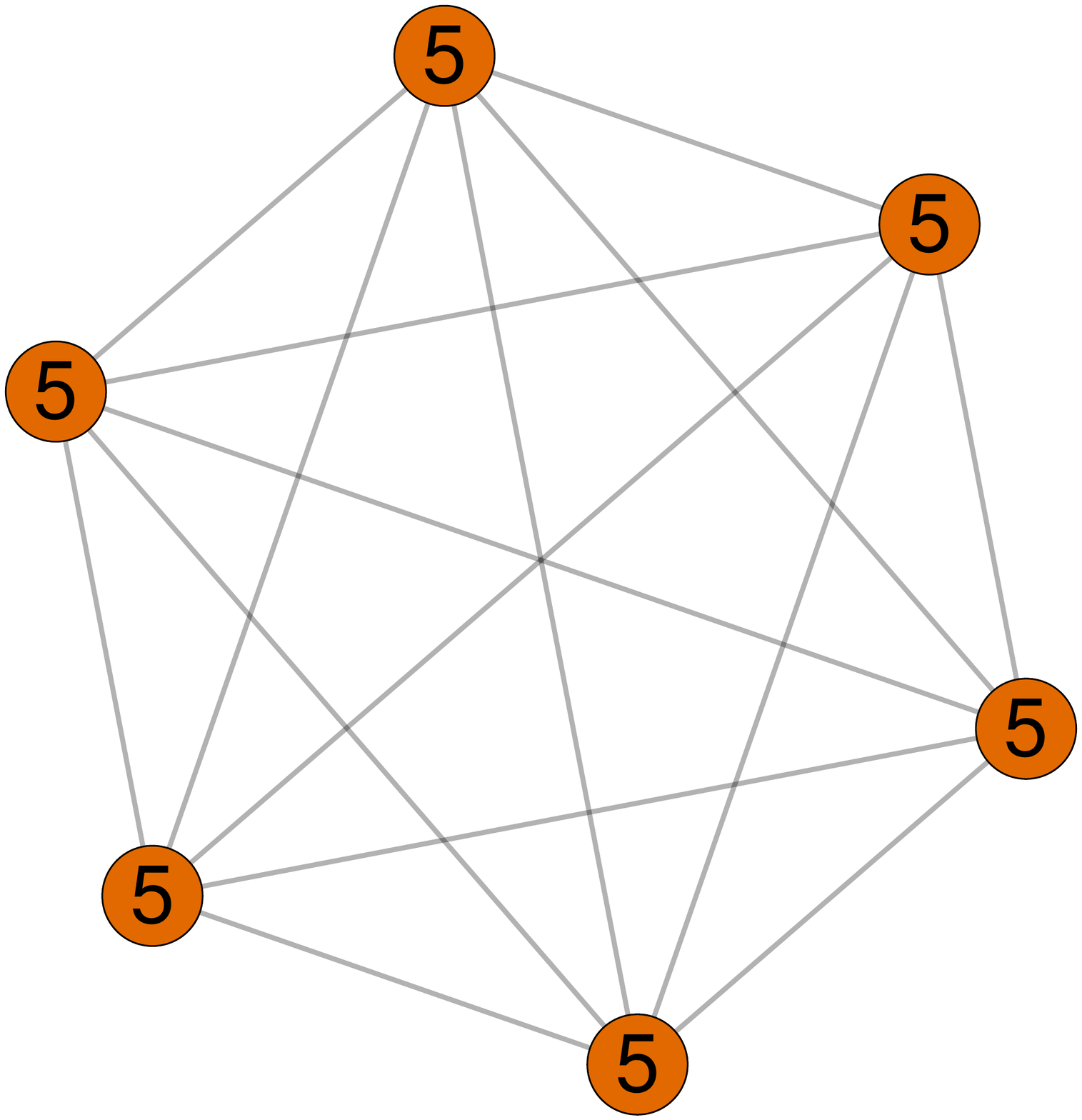}
        \caption{Protein}
    \end{subfigure}%
    \caption{\label{fig:commvis}Visualization of the subgraph formed using the innermost core nodes. Here nodes belonging to the same community are annotated with the same color and id. In the software network the innermost core clearly has nodes from several communities. In the protein network all nodes in the innermost core belong to the same community.}
  \end{figure}

\subsection{Evidence of rich centrality club} These motivating experiments demonstrate the existence two groups of networks.  One group consists of networks where the vertices from the inner cores have high correlation with vertices of different high centrality metrics and can be used as seed nodes for community detection. The other group consists of networks where the vertices from the inner cores have no correlation with other high centrality metrics, and are concentrated in one or two communities.  

To visualize how rich centrality clubs are formed in the innermost cores, we divide the vertices in the network into  two sets of nodes, based on whether they are part of the innermost core or not. We partition the nodes outside the innermost core into their respective communities and combine each community into a single supervertex. The set of nodes of  the innermost core are also combined into a supervertex. Two supervertices are connected if there is at least one edge between the nodes comprising them. 

\begin{figure}[t!]
    \centering
    \begin{subfigure}[t]{0.25\textwidth}
        \centering
        \includegraphics[height=1.4in]
        {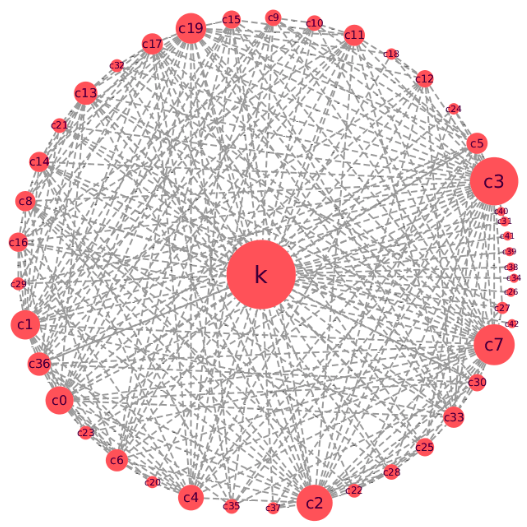}
        \caption{Caida}
    \end{subfigure}%
    ~
    \begin{subfigure}[t]{0.25\textwidth}
        \centering
        \includegraphics[height=1.4in]
        {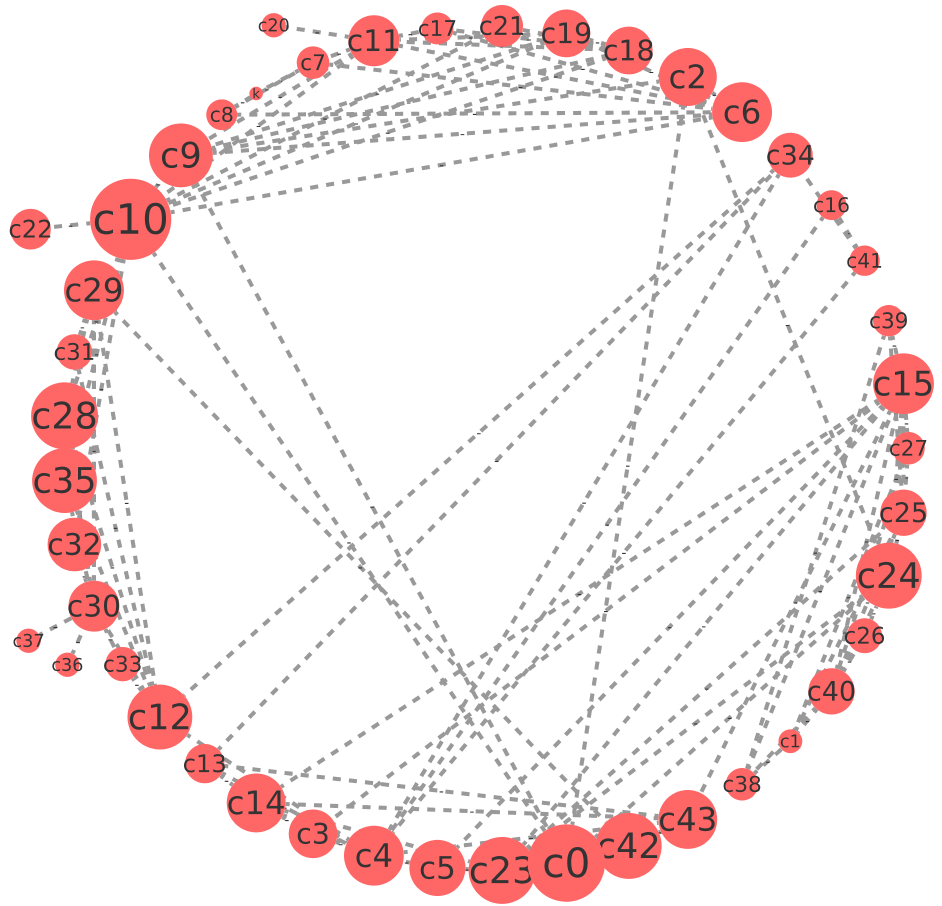}
        \caption{Power}
    \end{subfigure}%
    \caption{\label{fig:superv}Network formed by two category of supervertices, i.e, communities (denoted by $c_1$, $c_2$, $\dots$) and the innermost core (denoted by $k$). Two supervertices are connected if the corresponding nodes from which they are formed are connected by at least an edge. Higher size of supervertex imply higher average centrality of constituent vertices}
  \end{figure}

Figure~\ref{fig:superv} shows a visualization of the reduced network for two benchmark networks. Each node is labeled as $c_x$ for communities and $k$ for the innermost core. The supervertices are ordered by size with respect to average centrality(closeness and betweenness) of constituent vertices.  For the network Caida, which is in the first group, the supervertex corresponding to the innermost core has significantly high centrality and is in the centre. For the network, Power, which was in the second group, there are no distinctively high centrality supervertex. 

Since in the first group, the high centrality nodes are in the innermost cores and since by their definition these  vertices are connected to each other, our experiments demonstrate that {\em rich club of high centrality vertices} is formed in these networks. In the next section, we present the topological property of these networks that lead to the formation of RCC.








\section{Properties of Networks Containing Rich Centrality Clubs}\label{sec:part2}
We define structural properties of networks containing RCC and present empirical results to support our definition. In section~\ref{sec:theory} we provide theoretical rationale for our definition.

\begin{figure*}[t!]
    \centering
    \begin{subfigure}[t]{0.3333\textwidth}
        \centering
       \includegraphics[width=\textwidth,height=3.5cm]{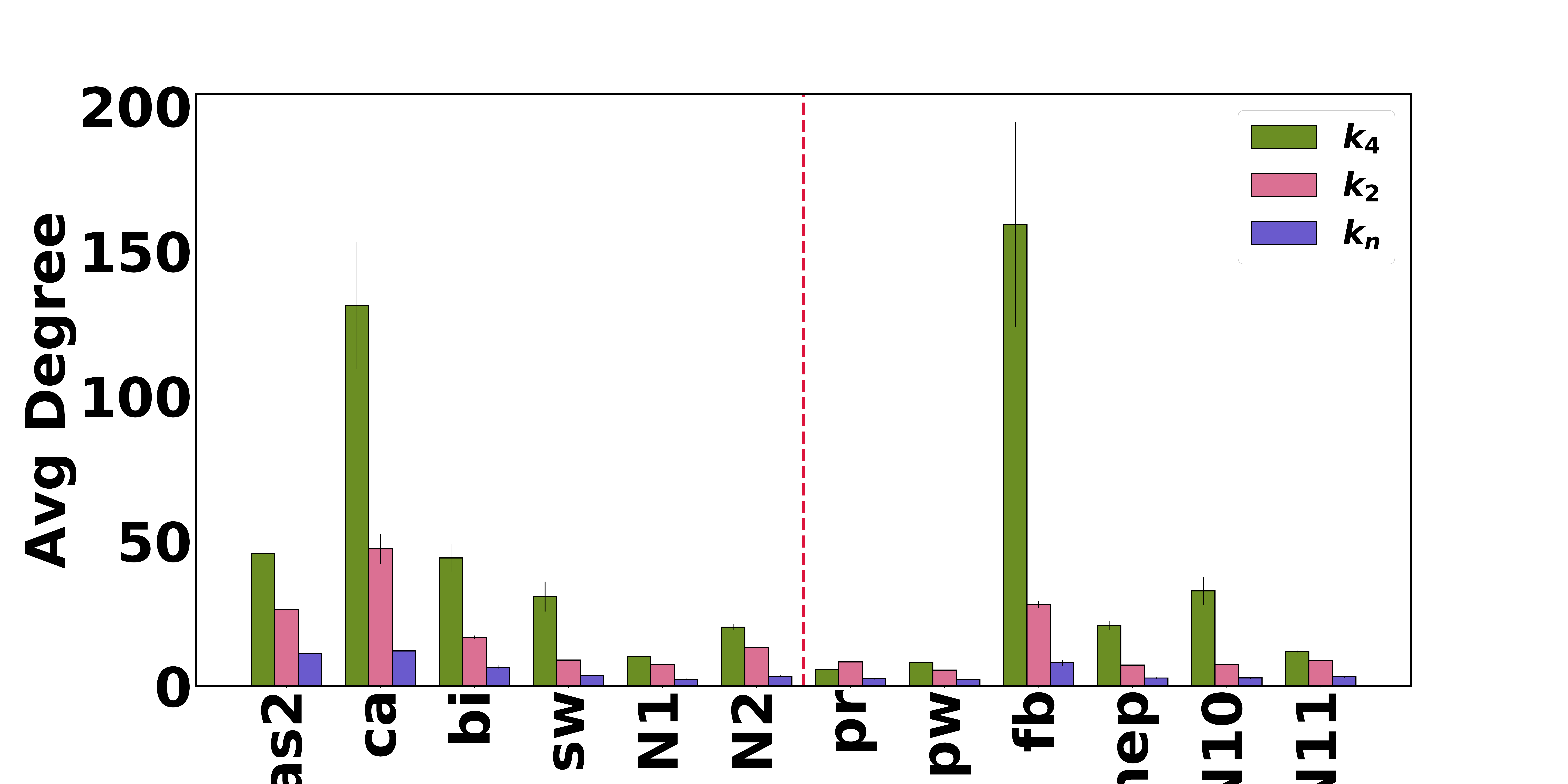}
        \caption{}
    \end{subfigure}%
    ~
    \begin{subfigure}[t]{0.3333\textwidth}
        \centering
     \includegraphics[width=\textwidth,,height=3.5cm]{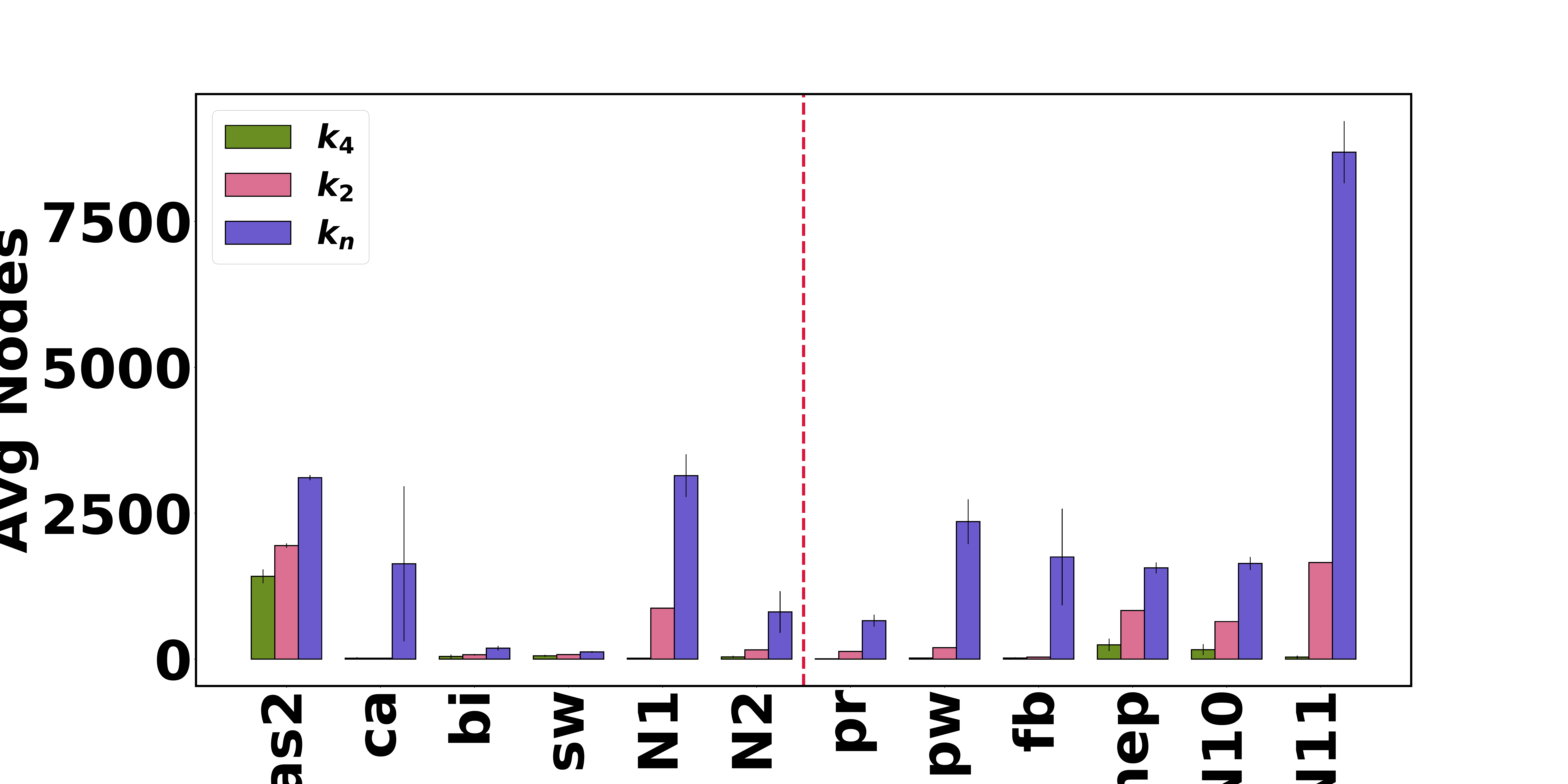}
        \caption{}
    \end{subfigure}%
    \begin{subfigure}[t]{0.3333\textwidth}
        \centering
     \includegraphics[width=\textwidth,height=3.5cm]{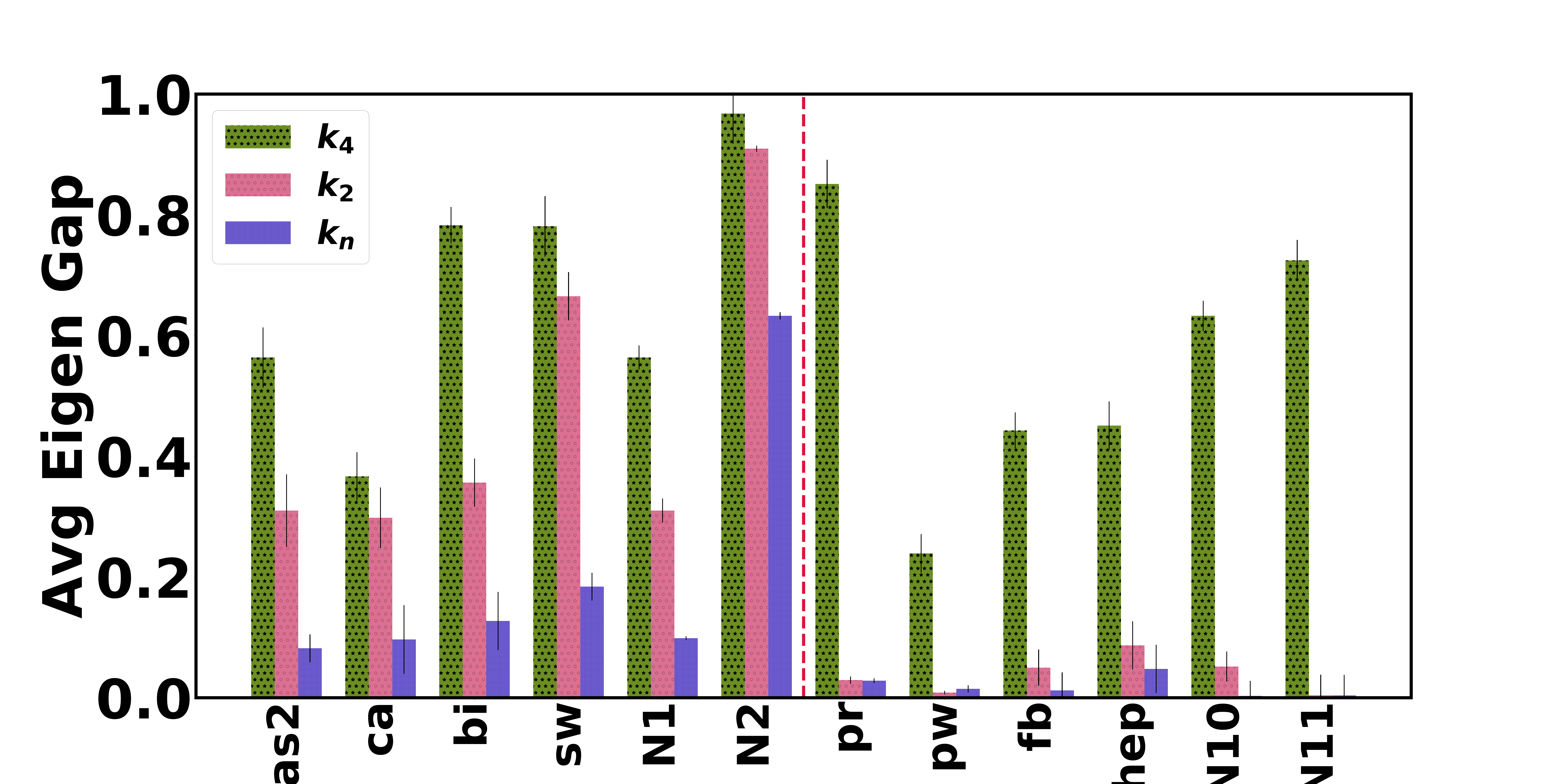}
        \caption{}
    \end{subfigure}%
    \caption{ \label{fig:shell_fig}(a) The average degree of, and (b) the number of nodes in, the shell based subgraphs for different buckets of shells for each network. (c) Average eigengap of the shell based subgraphs for different buckets of shells for each network. Results show graceful degradation for the networks with an RCC while an abrupt fall for the networks with no RCC.}
    
  \end{figure*}

\subsection{Formal definition and rationale}
Let the subgraph induced by the vertices in shell $k$ and their neighbors be  $S_{k}$. Let $d_k$ be the average degree  and $n_k$, the number of nodes in $S_{k}$. Let $\lambda_{k}$ be the  second smallest eigen value of the normalized Laplacian matrix of $S_{k}$. Let the average distance of a vertex in shell $S_a$ to a vertex an inner shell $C_{b}$, $a < b$, be $r_{ab}$. 

Given these parameters, we state that a network will contain a RCC if the following properties hold.
\begin{enumerate}
\item  If for two shells $k1$ and $k2$, $k1 < k2$, then $d_{k1} < d_{k2}$ and $n_{k1} > n_{k2}$.
\item For all shells $S_{k}$, $\lambda_{k} > \alpha$
\item For all shells $S_{k}$, $r_{kx} < \beta$, where $C_x$ is a high numbered core, with density close to 1. 
\end{enumerate}

The first property requires that the shells have progressively smaller number of vertices, and become more dense from outer to inner shells.  The second property provides the upper bound of the second smallest eigenvalue, and in turn, to the Cheeger constant at each shell. The higher $\lambda_k$, the more expander-like the associated shell. If the shell has multiple components, then each of them should maintain this property. The third property states that the hops to travel from the outer shells to inner cores should be small. 

The values of the parameters $\alpha$ and $\beta$ are determined based on size and density of the whole network. As per our experiments, setting $\alpha >.5$ and $\beta <4$ can clearly distinguish between networks that contain RCC and  those that do not.


\subsection{Density of shells}

The first condition is a feature of the core-periphery structure of almost all scale-free networks, whether they contain RCC or not. To demonstrate this, we first subdivide the vertices into subgraphs $S_k$. For each, we compute the average degree and the number of nodes. Since the number of shells varies across networks, for uniform presentation of the results we divide our results into three buckets. Starting from innermost we place the first $25\%$ shells in the first bucket($k_4$). The next $25\%$ falls in the second bucket($k_2$) and the final $50\%$ falls in the last bucket($k_n$). For each bucket we calculate the mean and the standard deviation of the average degrees and number of nodes classified in that bucket. These values are plotted in Figure~\ref{fig:shell_fig}(a) (average degree) and (b) (number of nodes). As seen from the figure with the exception of slight deviations, the first property is maintained in both sets of networks.

\subsection{Eigenvalue of shells} 

For each $S_k$, we compute the normalized Laplacian (see section~\ref{sec:prelims}), extract its spectrum using eigenvalue decomposition, and compute the eigengap. \if{0}Since the number of shells varies across networks, for uniform presentation of the results we divide our results into three buckets. Starting from innermost we place the first $25\%$ shells in the first bucket. The next $25\%$ falls in the second bucket and the final $50\%$ falls in the last bucket.\fi For each bucket as defined  in the previous section, we calculate the average eigengap and the standard deviation. These values are plotted in Figure~\ref{fig:shell_fig}(c).

We observe that in graphs where we assume that RCC exists, there is a slow decline of the average eigengap. The Cheeger constant is high in the inner shells and gradually decreases from the inner to the outer shells. In the other group of networks, there is an abrupt fall in the average eigengap after the first bucket of inner shells. The first group can be bound by a large $\alpha$ than the second group, thus corroborating the second property.


\subsection{Distance between shells}
The third property enforces that on average two vertices in outer shells are more likely to be connected through inner dense shells. We show this in Figure~\ref{fig:shell_dist}, where in networks with an RCC, the average shortest distance of the nodes in the outer shells to the innermost ($k_{max}$) and the second innermost shells ($k_{max-1}$) is low (2-3) compared to networks without an RCC (10-50). 

\begin{figure}
\centering
\includegraphics[scale=0.08]{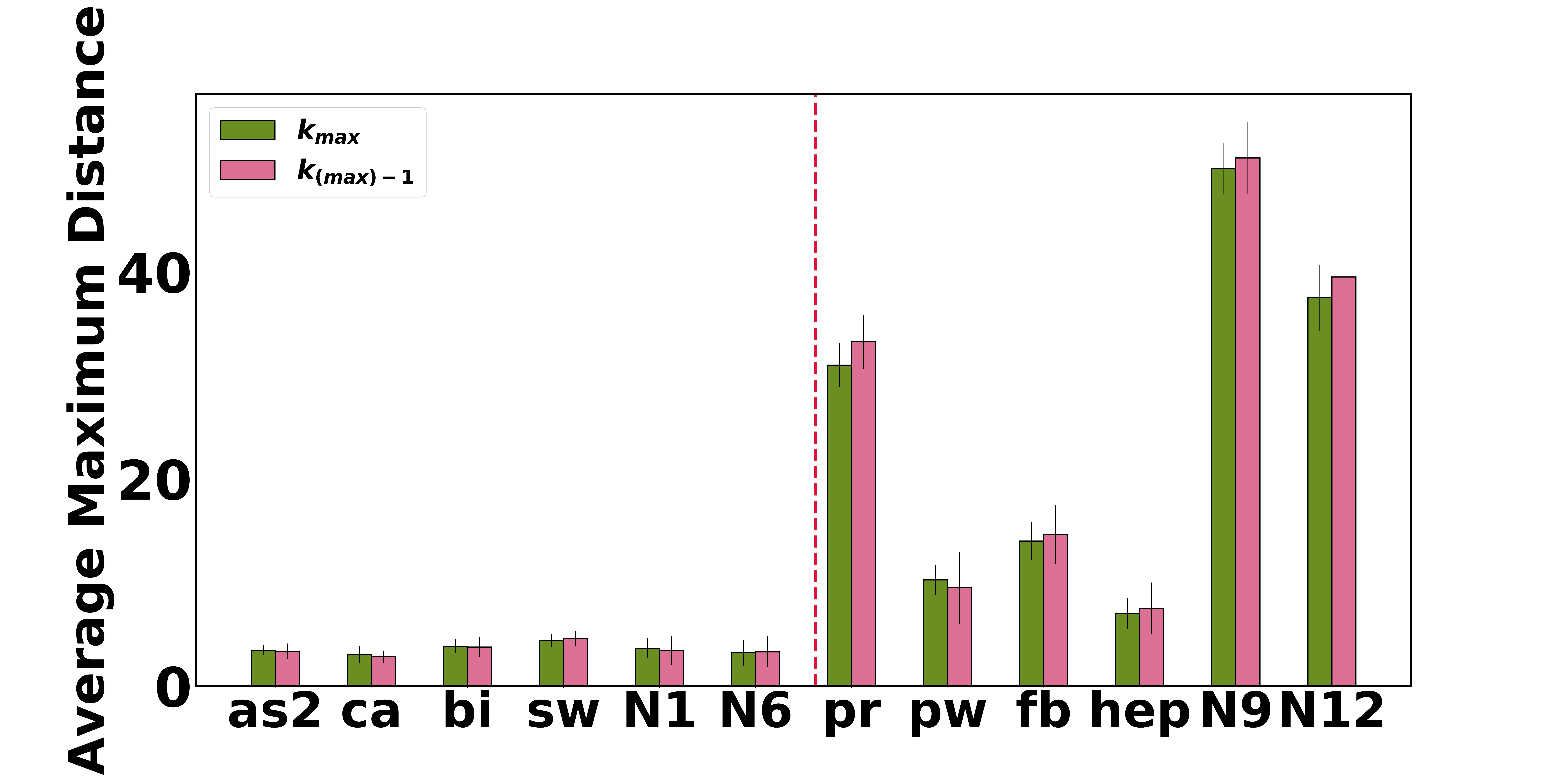}
\caption{\label{fig:shell_dist} The average shortest distance of a node in the outer shells to a node in the innermost ($k_{max}$) and the second innermost shell ($k_{max-1}$)}
\end{figure}

\section{Application}\label{sec:application}
In this section, we demonstrate how the presence of RCC can be leveraged in some important applications.
\subsection{RCC as influential nodes}
\label{sec:spreads}
Vertices, which when selected as seed nodes can accelerate the diffusion process in networks are known as influential nodes. 
We hypothesize that the RCC, if present in a network, is a natural choice for the seed nodes for spreading. In order to test this hypothesis we execute a diffusion process adapted from~\cite{maiya2014expansion} on two groups networks in our dataset, i.e., those with an RCC and those without. 


We choose a seed set size of  $\mathit{S}$\footnotetext{Our experiments with different value of $\mathit{S}$ yield similar results.} (here we show results for $\mathit{S}=10$) and populate this set preferentially on the basis of highly connected nodes which includes high centrality nodes (degree, closeness, betweenness), innermost shell nodes (these are nodes from within the RCC in networks that demonstrate its presence) and a random set of nodes. These set of seed nodes initially have an information and they pass it to all their neighbors using a flooding based broadcast approach (all neighbors of an informed vertex get informed).
This approach spreads the message very quickly and hence modified versions have been used in peer-to-peer networks for sending queries and searching \cite{jiang2008lightflood}. Although in real world systems this method is difficult to implement due to scalability issues, our goal here is to study how effective the vertices in the RCC are as seed nodes.
 
The x-axis in Figure~\ref{fig:spread} shows the fraction of vertices that have received the message and the y-axis the steps to reach these fraction of vertices. The networks form two groups. In one group that demonstrate the presence of RCC, the vertices from the innermost core are effective seed nodes for broadcasting and the time is comparable to the time when high centrality vertices are selected as seeds. The other group is when the vertices from the innermost core perform very poorly as seed nodes. The time to spread the information is equal to or worse than a random selection. These results show that {\em only vertices in the innermost core in networks that demonstrate the presence of RCC are effective as seed nodes for spreading information}.

\begin{figure*}
  \begin{subfigure}[b]{0.15\textwidth}
    \includegraphics[width=\textwidth]{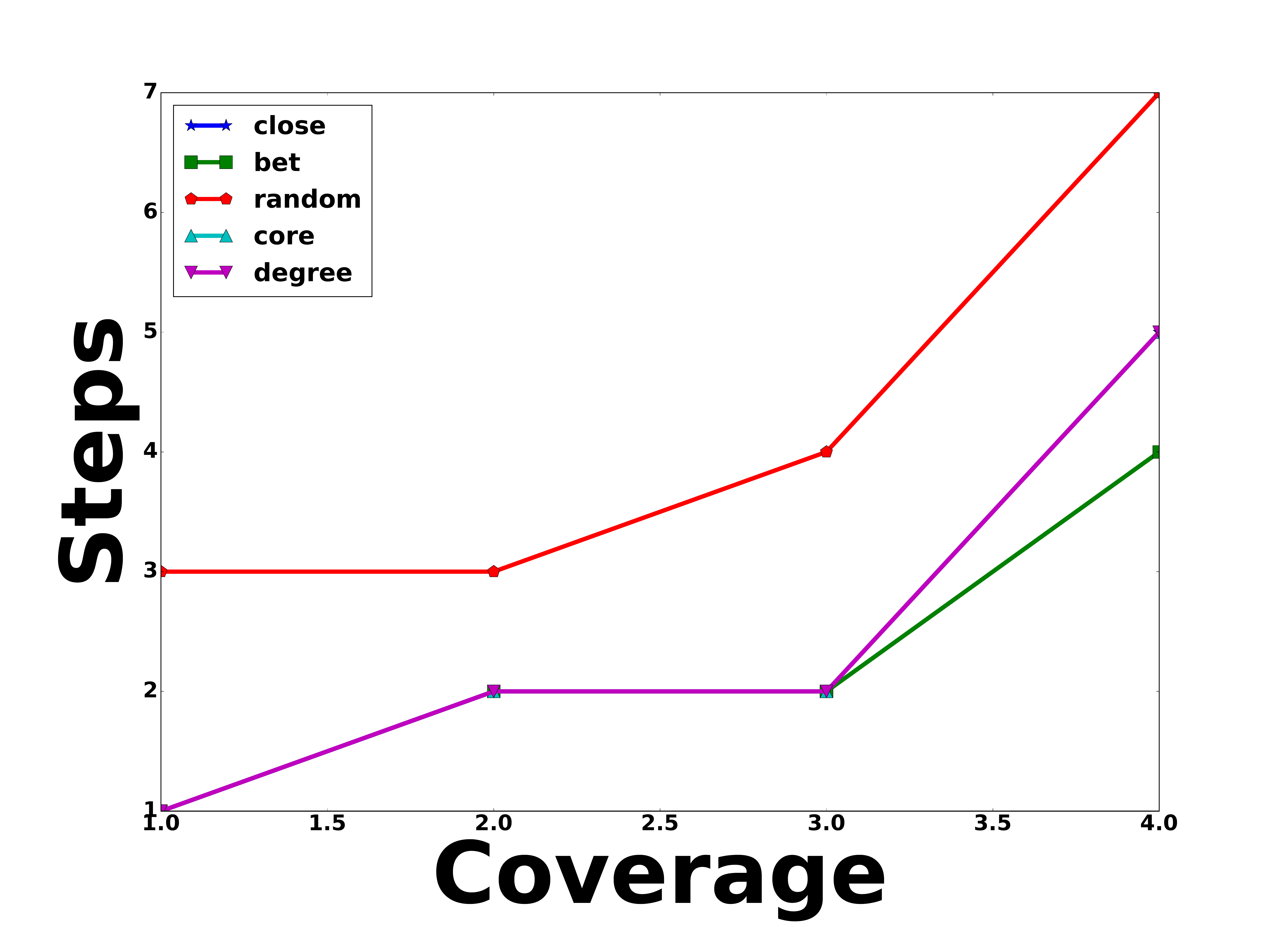}
    \caption{AS}
  \end{subfigure}
  \begin{subfigure}[b]{0.15\textwidth}
    \includegraphics[width=\textwidth]{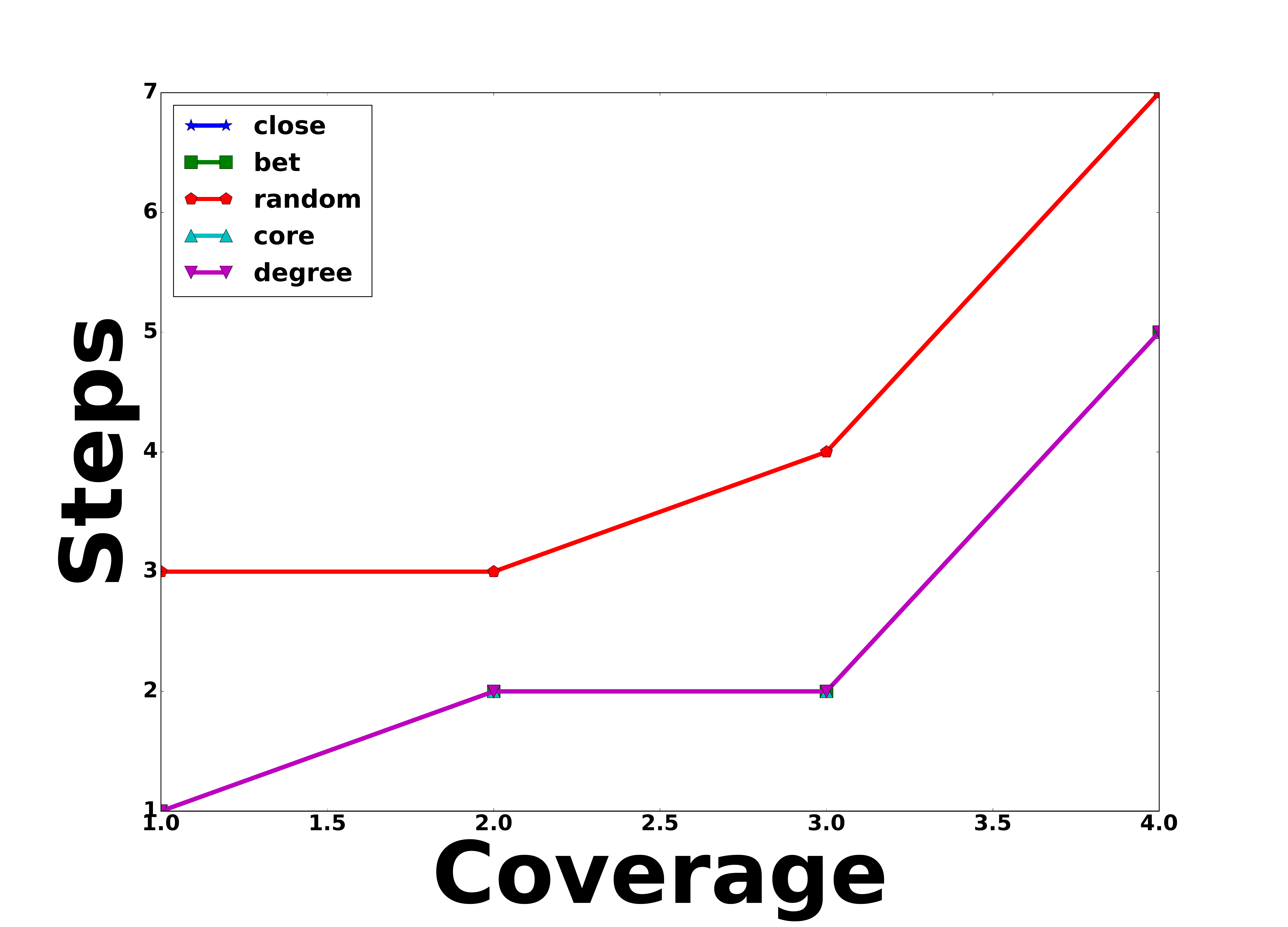}
    \caption{Caida}
  \end{subfigure}
  \begin{subfigure}[b]{0.15\textwidth}
    \includegraphics[width=\textwidth]{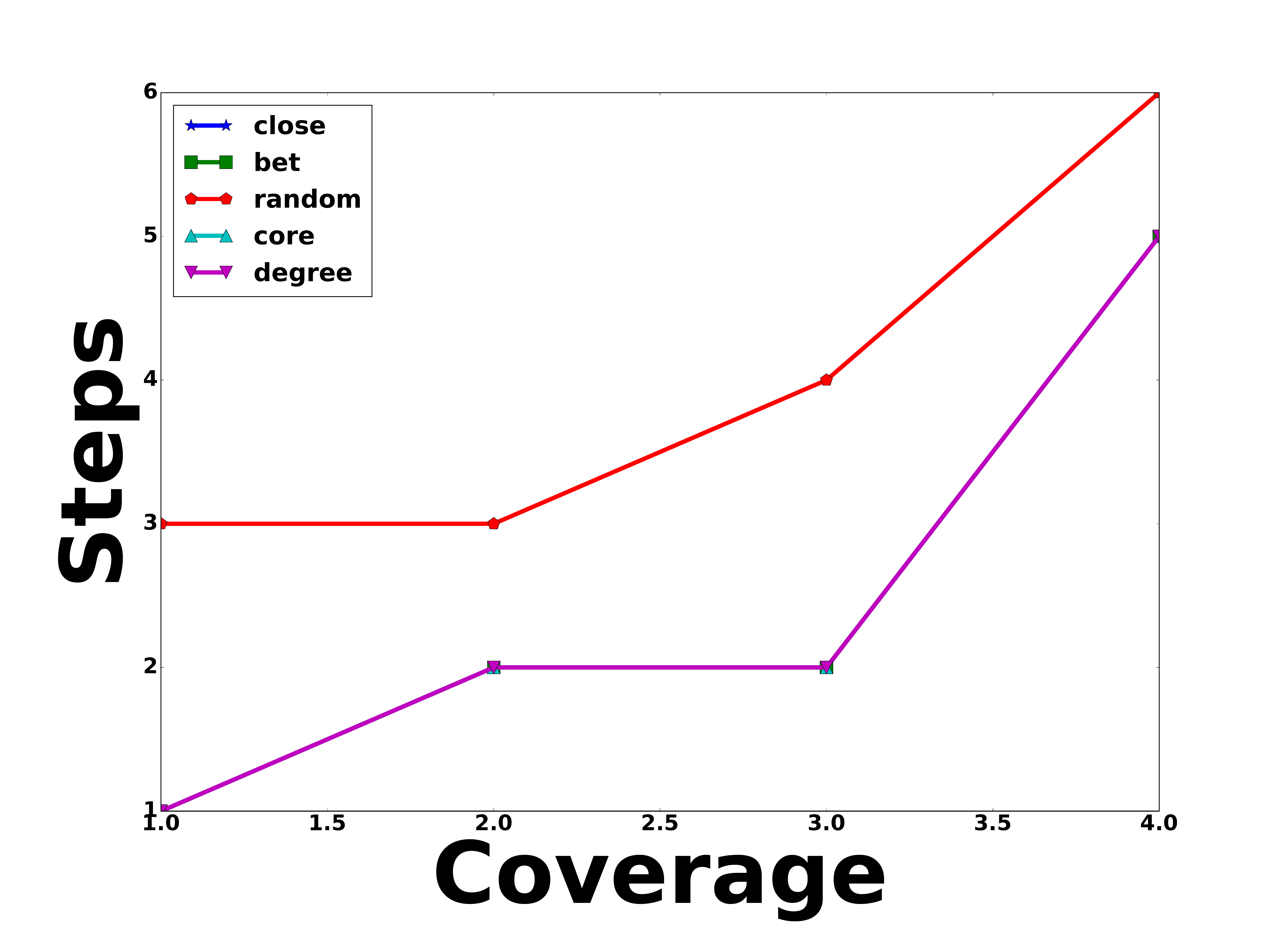}
    \caption{Bible}
  \end{subfigure}
  \begin{subfigure}[b]{0.15\textwidth}
    \includegraphics[width=\textwidth]{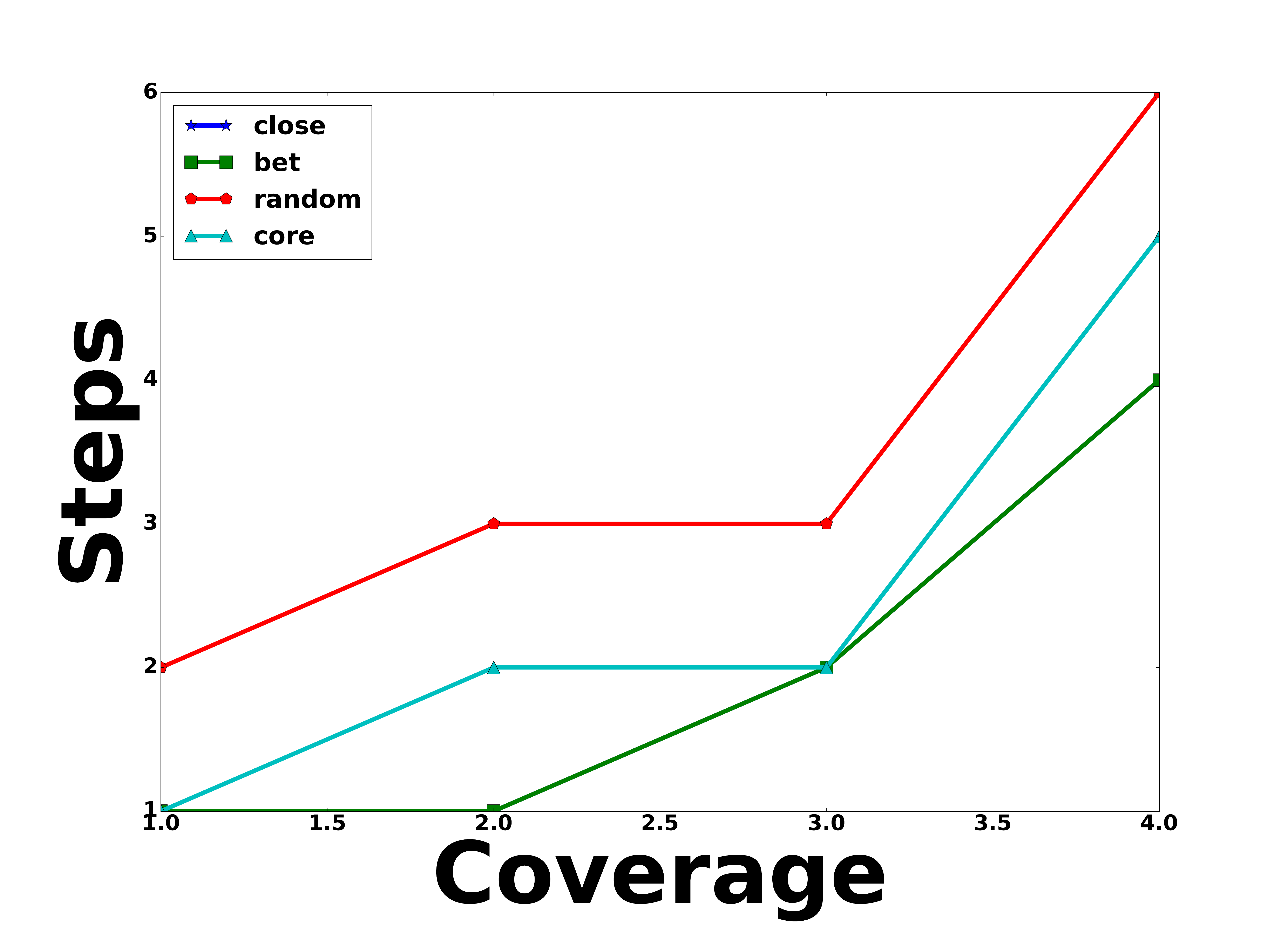}
    \caption{Software}
  \end{subfigure}
  \begin{subfigure}[b]{0.15\textwidth}
    \includegraphics[width=\textwidth]{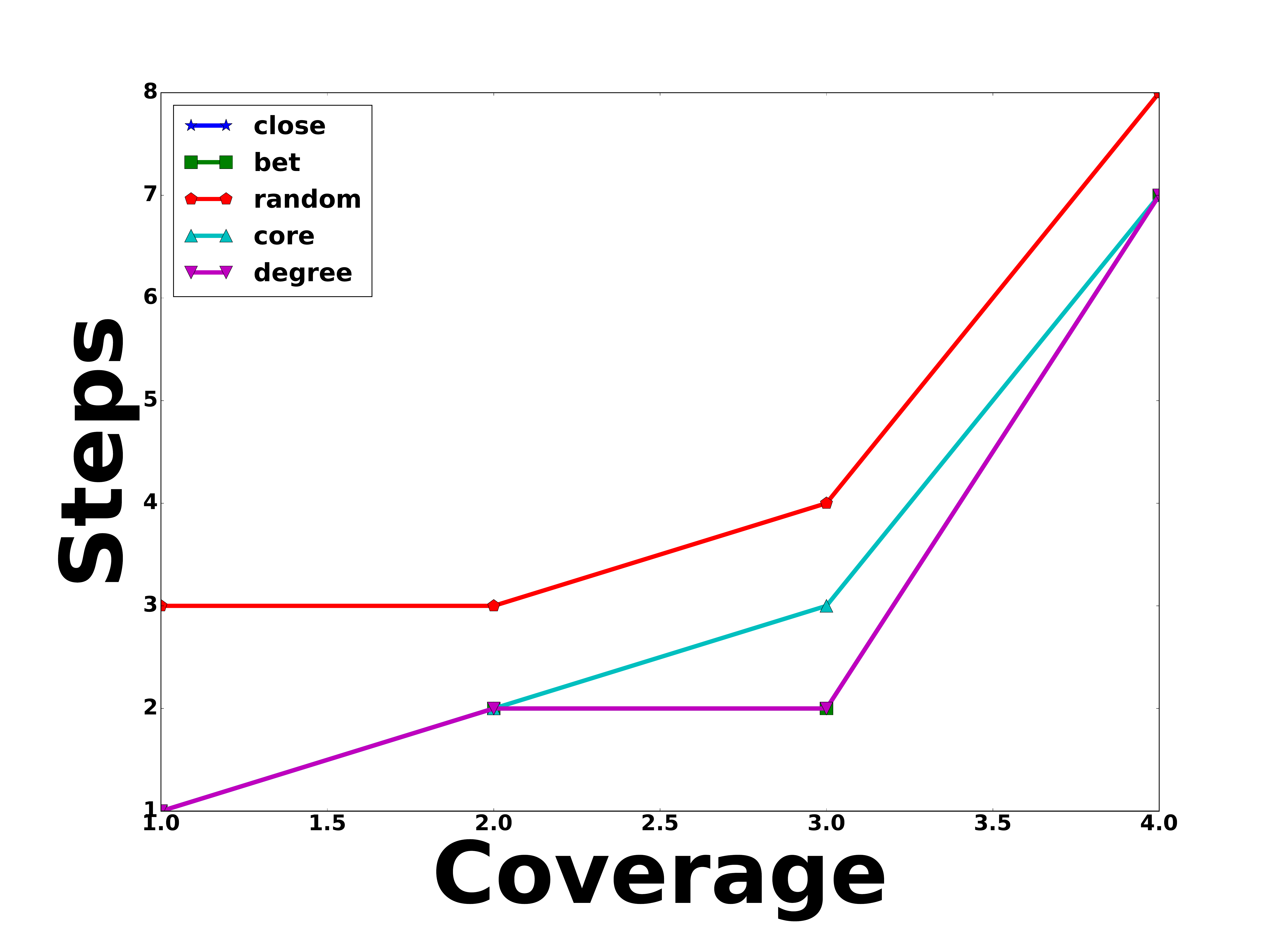}
    \caption{N1}
  \end{subfigure}
  \begin{subfigure}[b]{0.15\textwidth}
    \includegraphics[width=\textwidth]{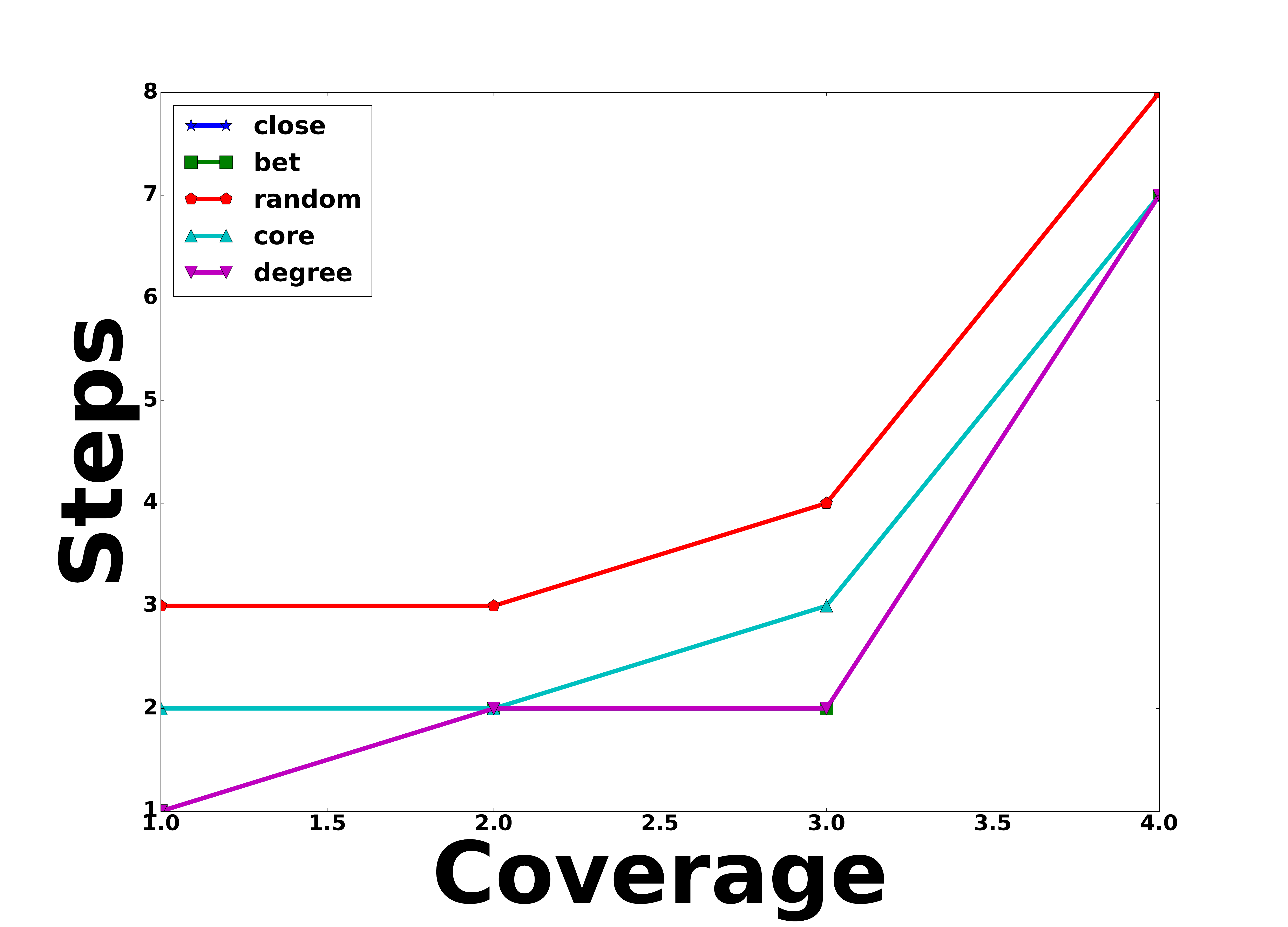}
    \caption{N3}
  \end{subfigure}

  \begin{subfigure}[b]{0.15\textwidth}
    \includegraphics[width=\textwidth]{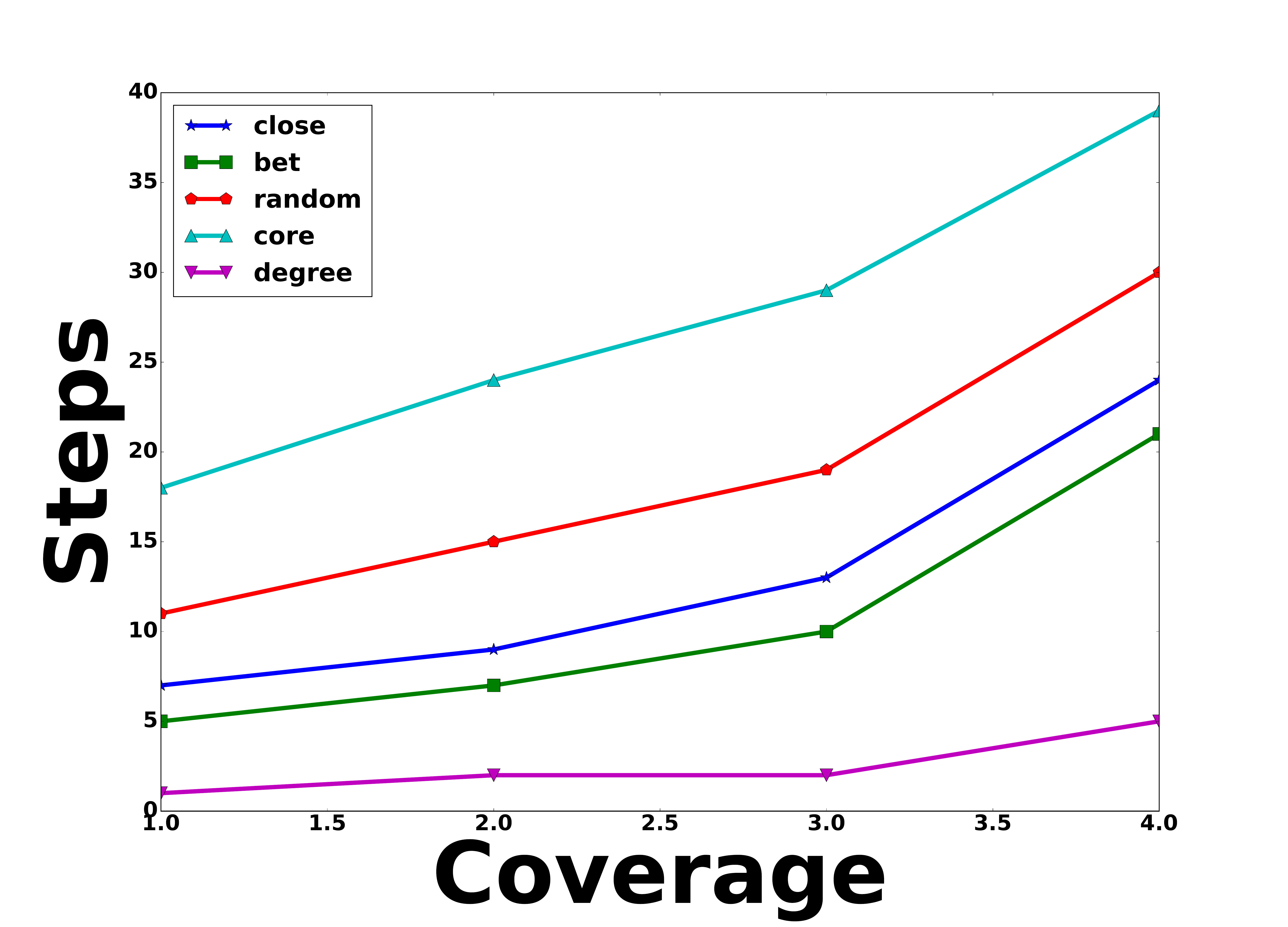}
    \caption{Power}
  \end{subfigure}
  \begin{subfigure}[b]{0.15\textwidth}
    \includegraphics[width=\textwidth]{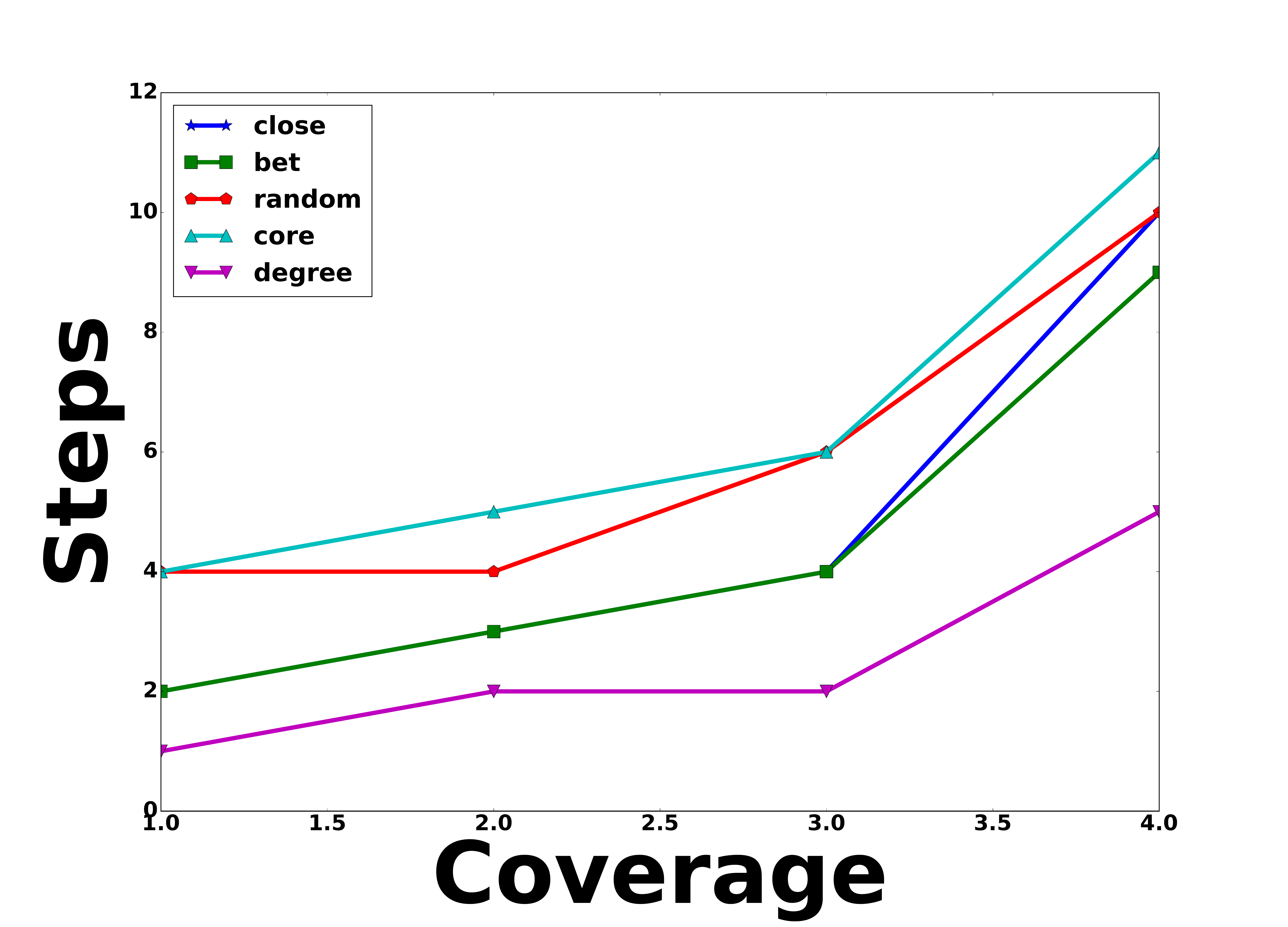}
    \caption{Protein}
  \end{subfigure}
  \begin{subfigure}[b]{0.15\textwidth}
    \includegraphics[width=\textwidth]{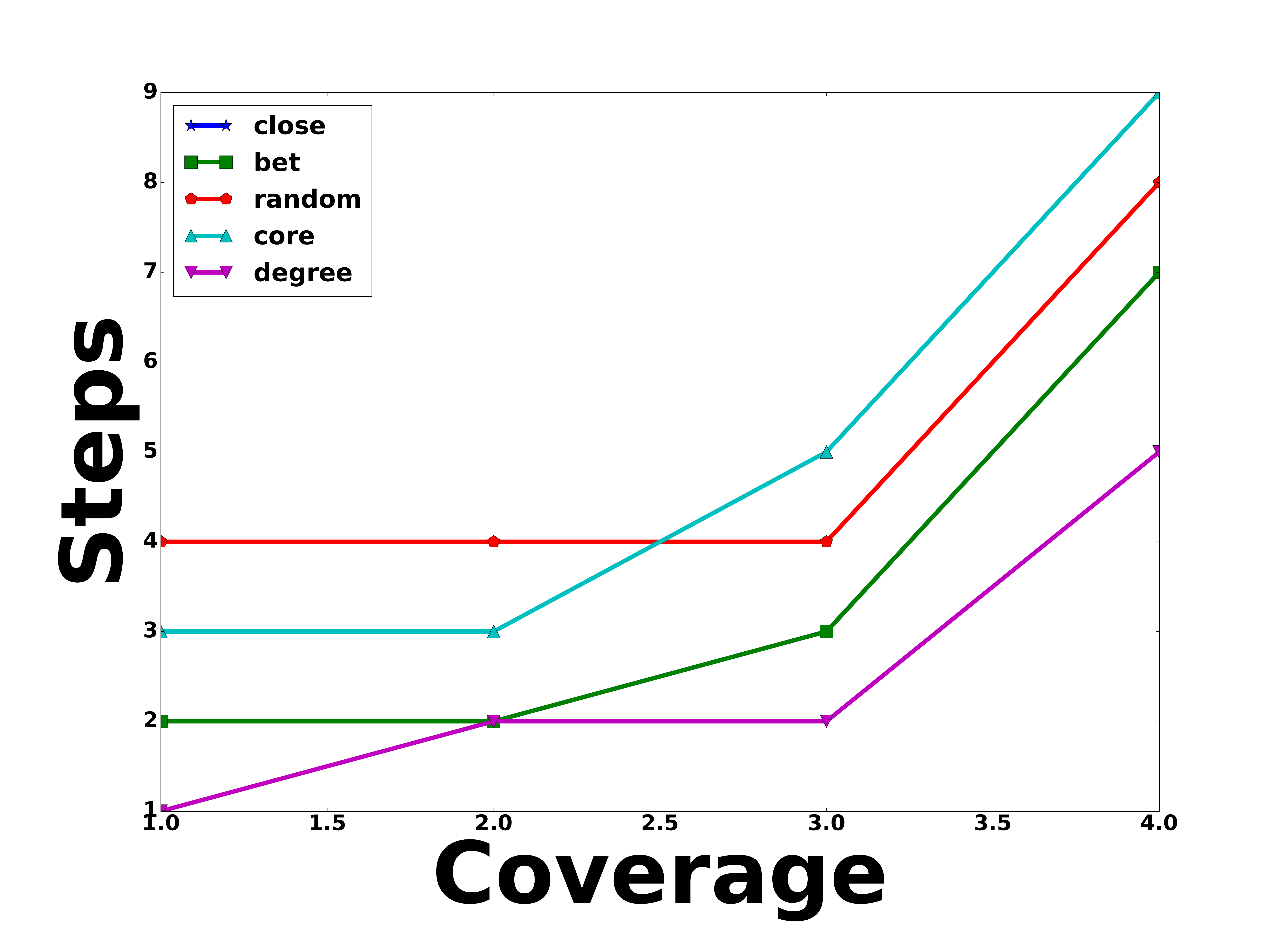}
    \caption{Hepth}
  \end{subfigure}
  \begin{subfigure}[b]{0.15\textwidth}
    \includegraphics[width=\textwidth]{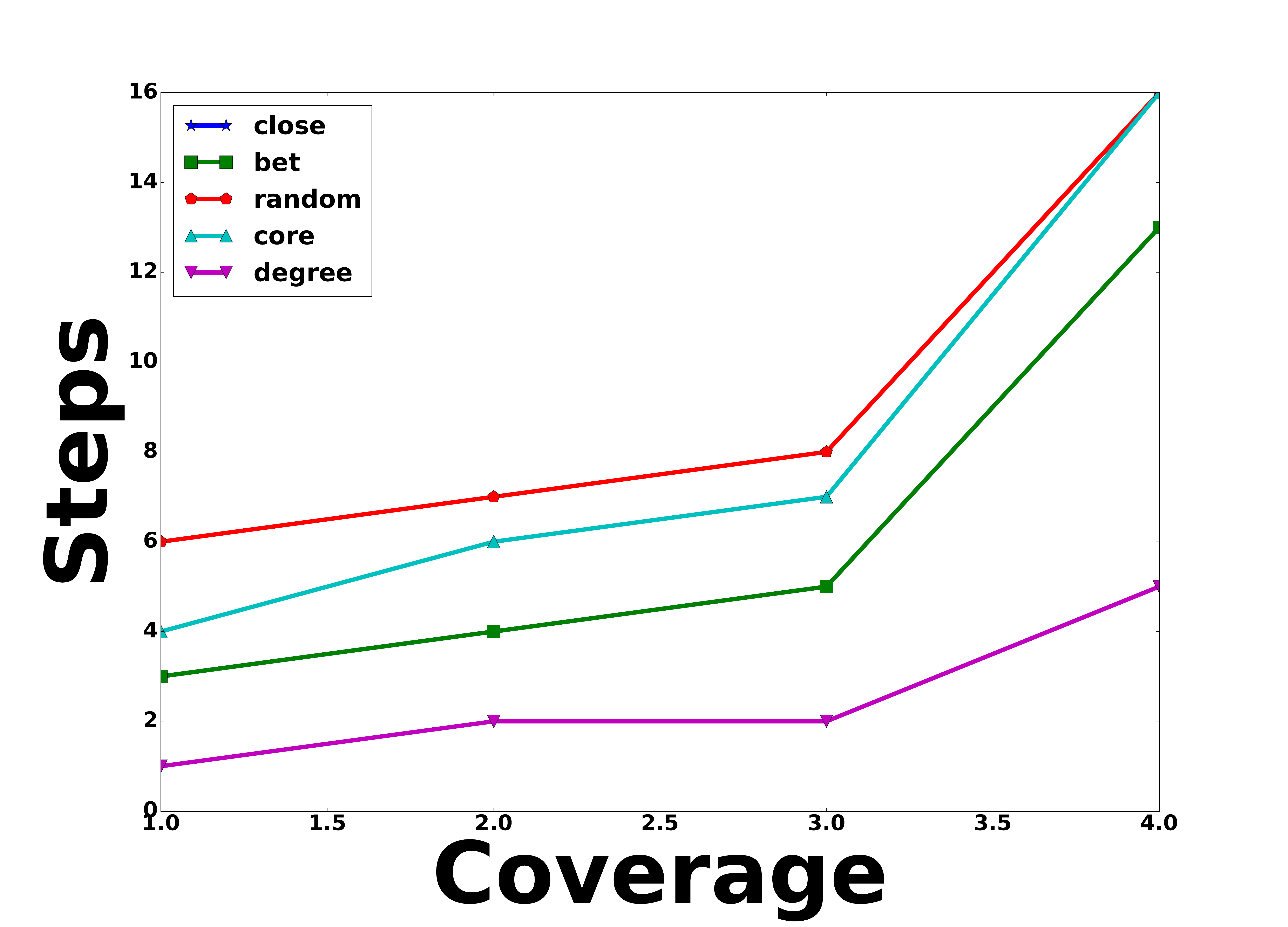}
    \caption{Facebook}
  \end{subfigure}
  \begin{subfigure}[b]{0.15\textwidth}
    \includegraphics[width=\textwidth]{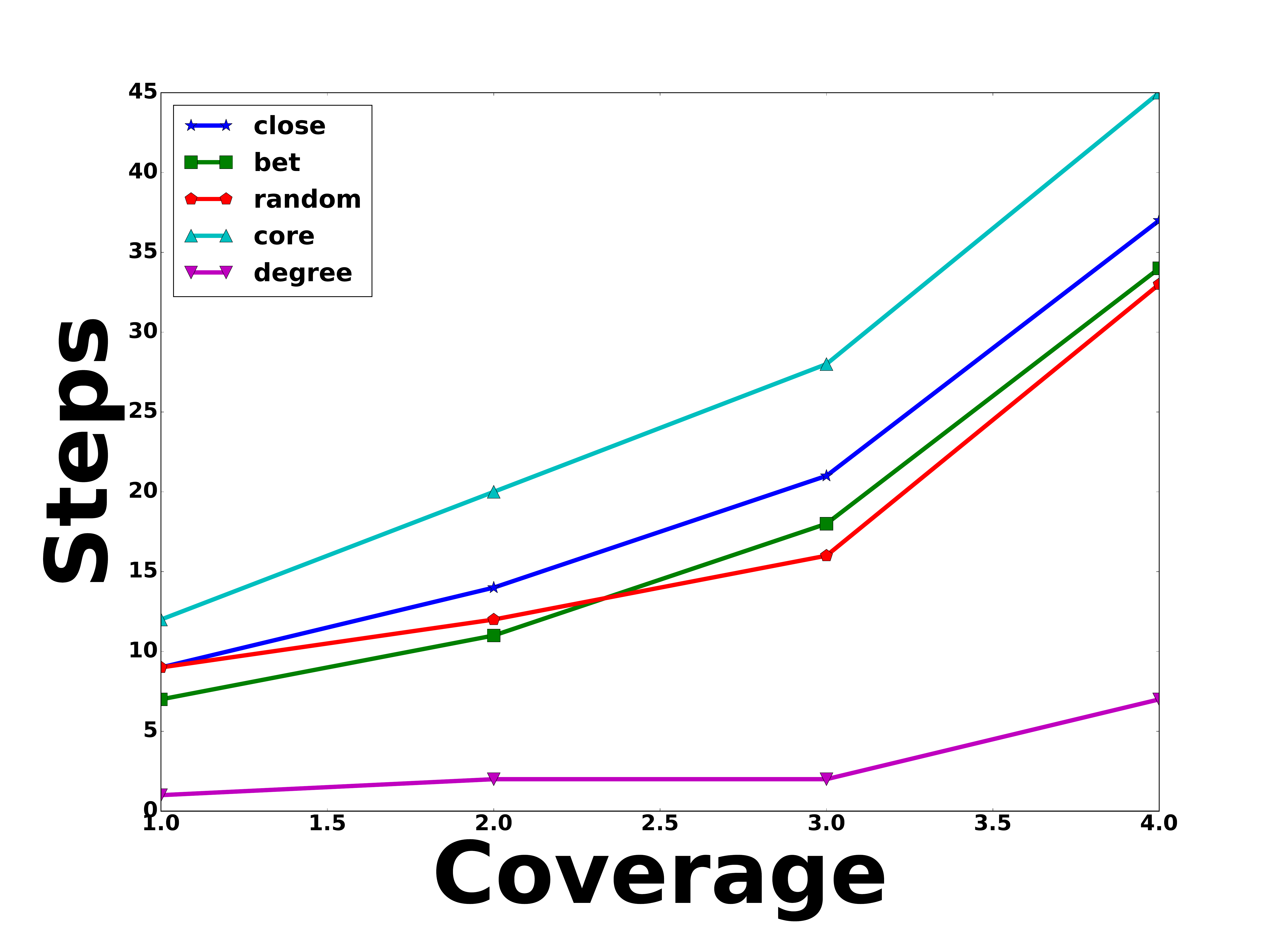}
    \caption{N8}
  \end{subfigure}
  \begin{subfigure}[b]{0.15\textwidth}
    \includegraphics[width=\textwidth]{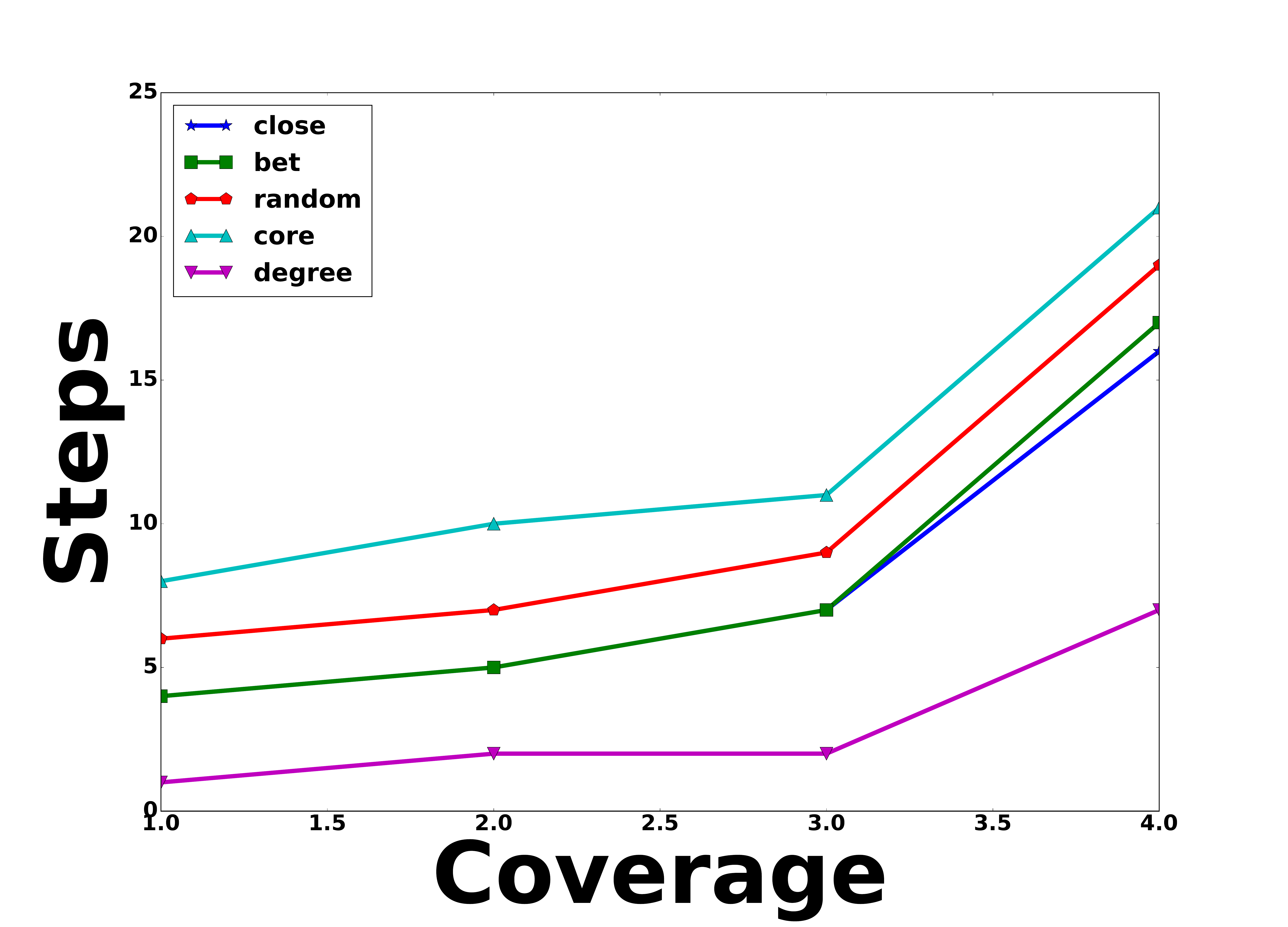}
    \caption{N9}
  \end{subfigure}
  \caption{\label{fig:spread}Time required in terms of the number of steps for the information to disseminate to $\frac{n}{4},\frac{n}{2},\frac{3n}{4},n$ nodes in the network. In the top panel, for the networks that demonstrate the presence of RCC, coreness based seed nodes consistently appear to be good choices as message initiators. In the bottom panel, for the networks that do not demonstrate the presence of RCC, coreness based initiators perform equal to or worse than random initiators.} 
    
\end{figure*}
\subsection{Robustness}
One of the desirable properties of a network is whether it can retain the ranking of its top-$k$ high centrality nodes under perturbation of the network. We hypothesize that networks which demonstrate the presence of RCC are robust to minor perturbations in the form of random deletion of edges.

To corroborate this we compute the closeness and the betweenness centrality based rankings of the nodes in the original network and compare the top-$k$ ranked nodes with that of the perturbed network. Perturbation is done by randomly deleting 1\% to 8\% of the edges from the original network. For comparison we use the standard Kendall $\tau$ measure as demonstrated in Laishram et. al.~\cite{laishram2018measuring}. As discussed in this work we also use $k=50$; however, experiments with $k=10, 20$ etc., also yield similar results. Results obtained for both categories of networks, i.e., those in which an RCC is present and those in which it is not are illustrated in Fig~\ref{fig:robust}.


Our results show that for minor perturbations of the topology, networks with RCC are robust in terms of preserving the top-$k$ high centrality nodes. However, if the perturbation intensity is too high the ranking gets jeopardized. In case of networks without an RCC even a small perturbation substantially disturbs the ranking of the top-$k$ high centrality nodes. 

\begin{figure*}[t!]
    \centering
    \begin{subfigure}[t]{0.25\textwidth}
        \centering
       \includegraphics[width=\textwidth, height=3.5cm]{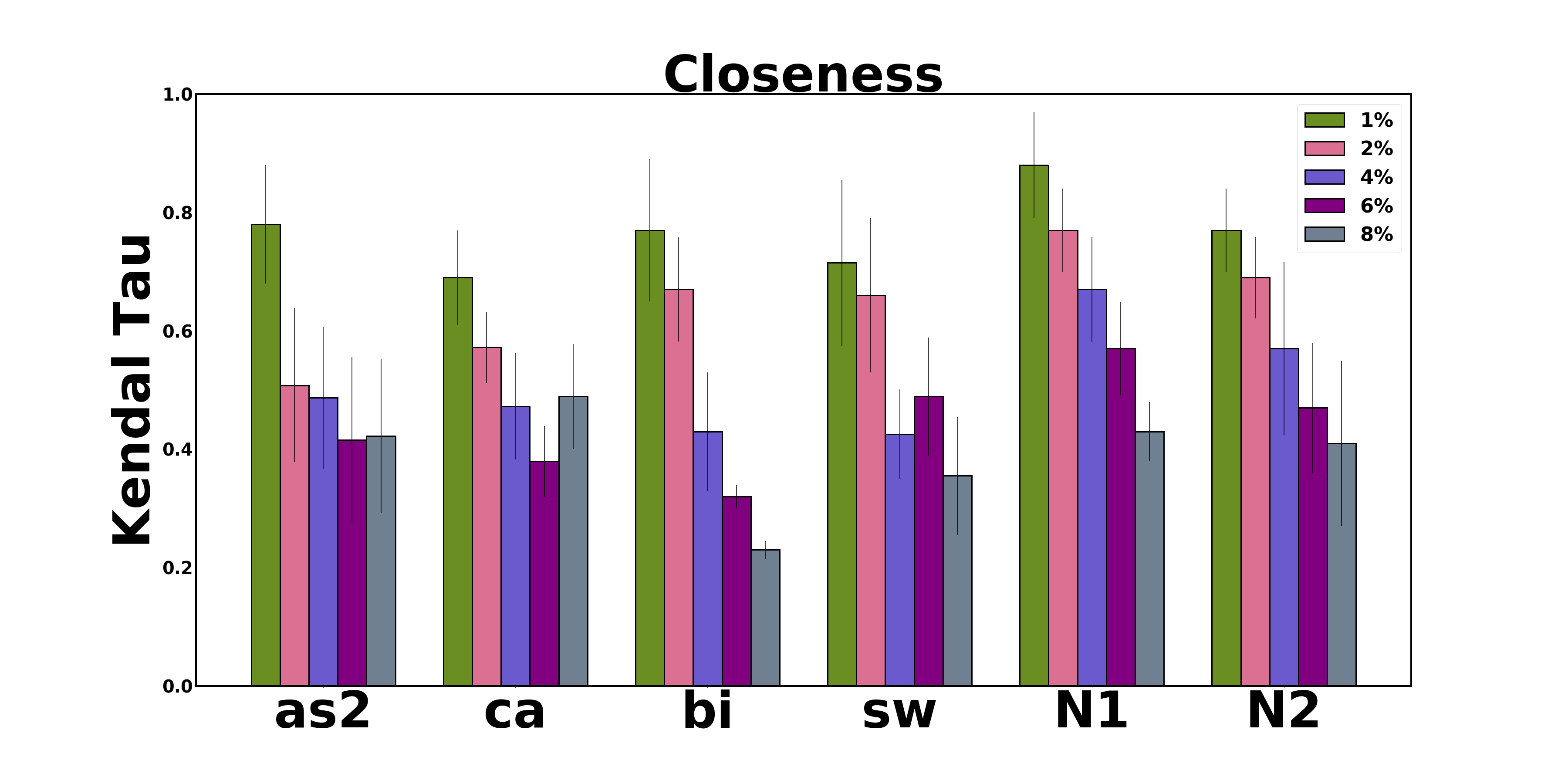}
        \caption{}
    \end{subfigure}%
    ~
    \begin{subfigure}[t]{0.25\textwidth}
        \centering
     \includegraphics[width=\textwidth, height=3.5cm]{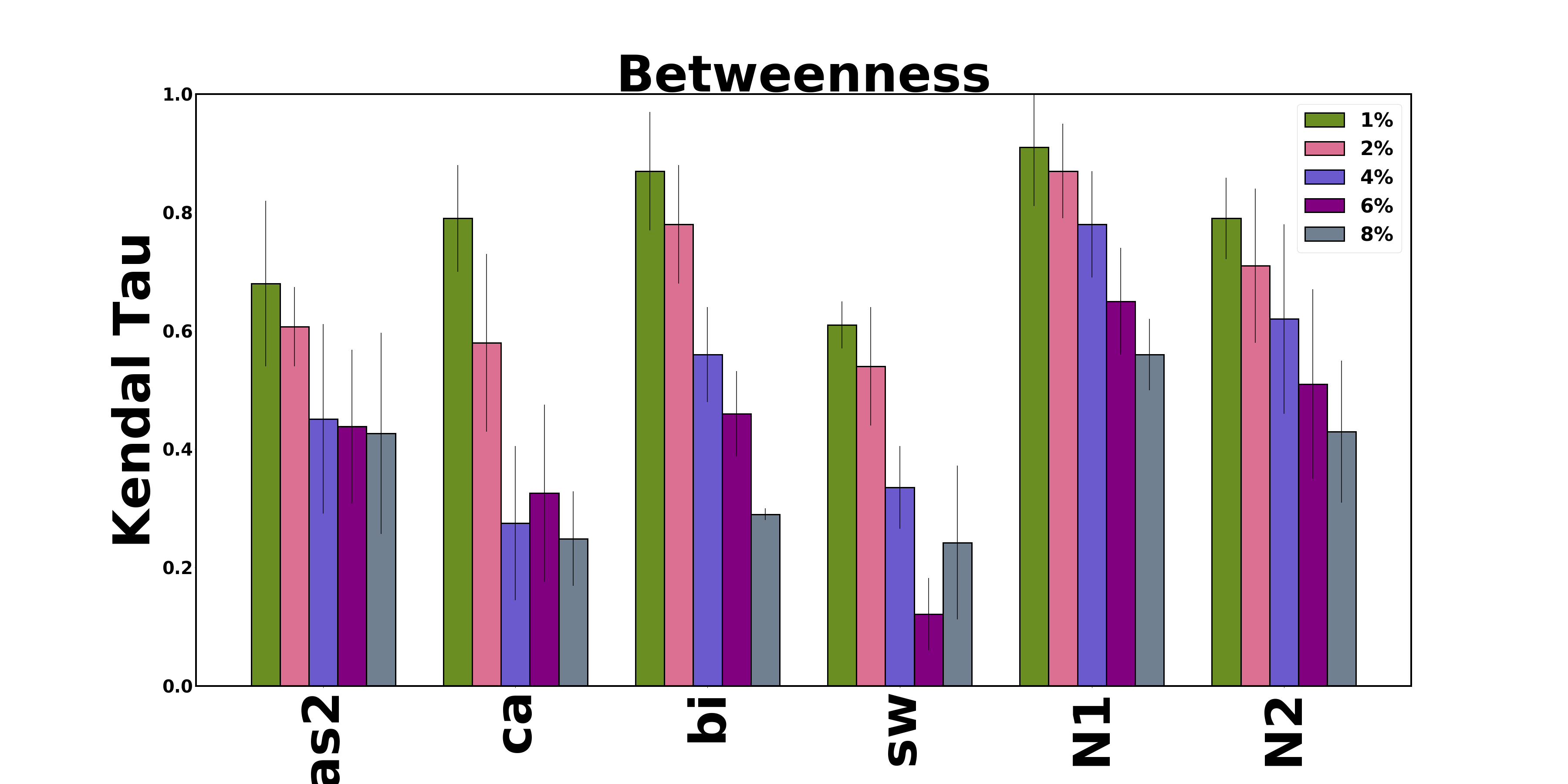}
        \caption{}
    \end{subfigure}%
    \begin{subfigure}[t]{0.25\textwidth}
        \centering
   \includegraphics[width=\textwidth, height=3.5cm]{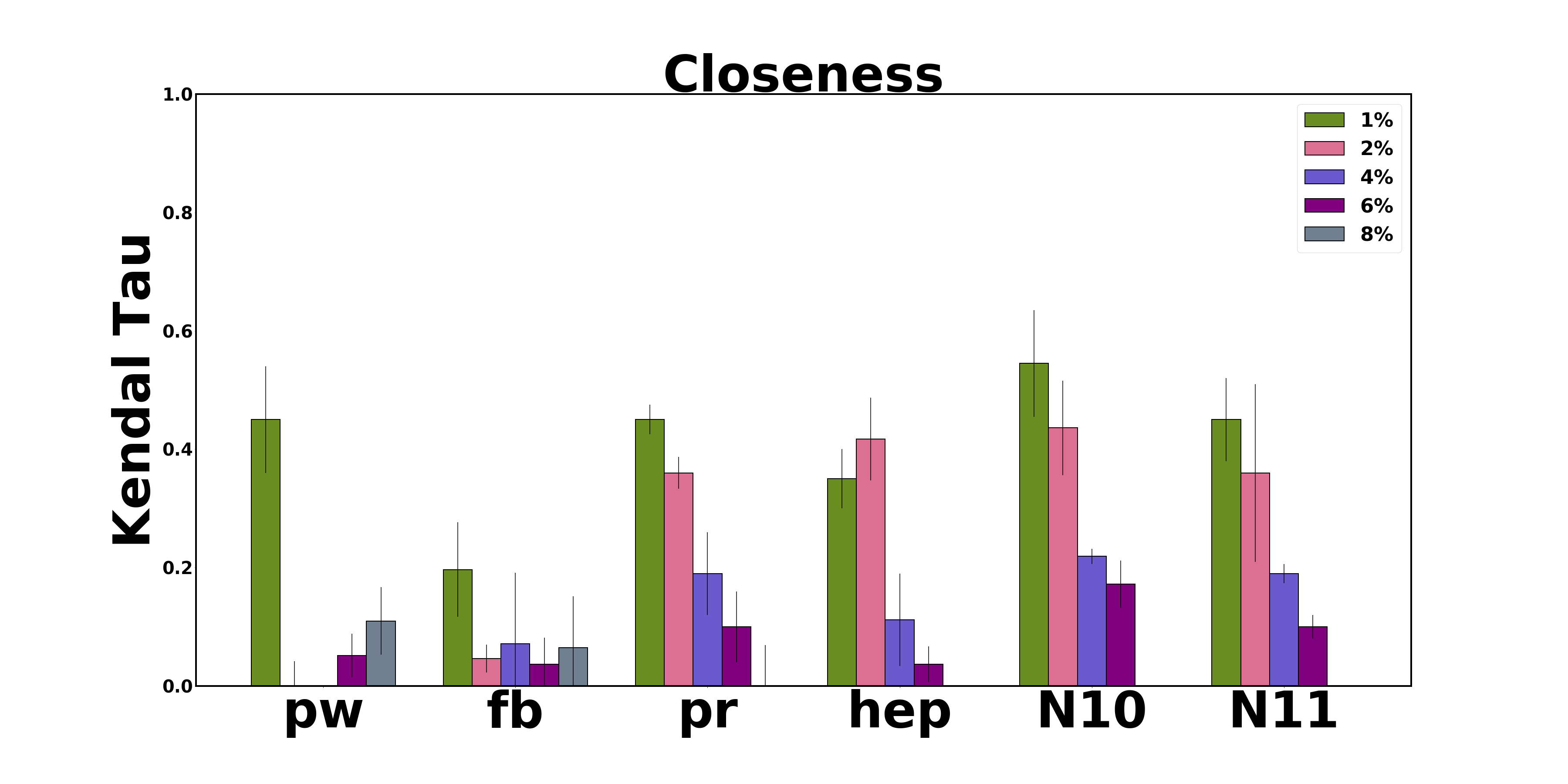}
        \caption{}
    \end{subfigure}%
    \begin{subfigure}[t]{0.25\textwidth}
        \centering
        \includegraphics[width=\textwidth, height=3.5cm]{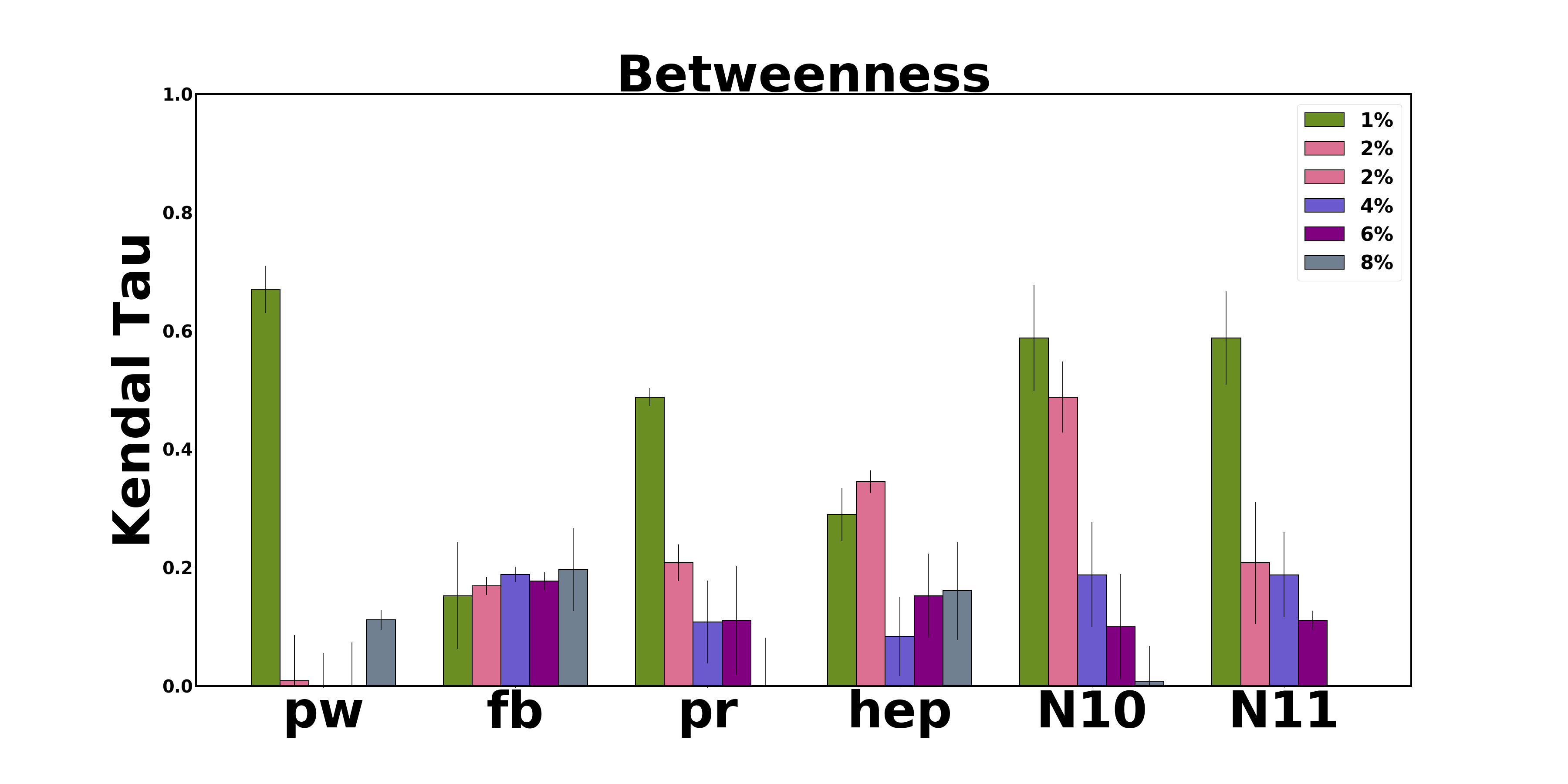}
        \caption{}
    \end{subfigure}%
    \caption{\label{fig:robust} Robustness results in terms of Kendal $\tau$ score comapring the rankings of the top-50 high closeness and betweenness nodes of the original and the perturbed version of the network. Perturbation is done by randomly deleting 1\% to 8\% of the edges. We average the results over 10 different runs; the error bars therefore are also reported. Results obtained in the case of networks with an RCC (5a, 5b) is compared against networks without an RCC (5c, 5d). The results clearly indicate that networks with RCC are robust to minor perturbations.}    
  \end{figure*}
\begin{figure}
  \centering
  \begin{minipage}[b]{\textwidth}
  	\includegraphics[width=0.25\textwidth, height=0.2\textheight]{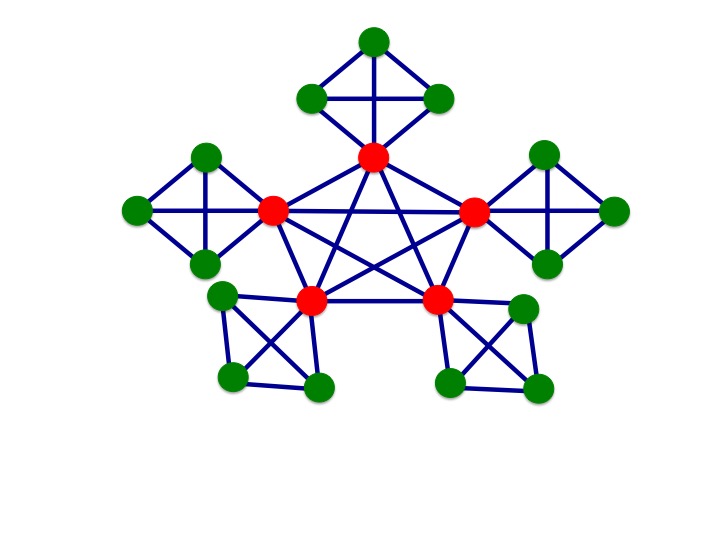}
     \includegraphics[width=0.25\textwidth, height=0.2\textheight]{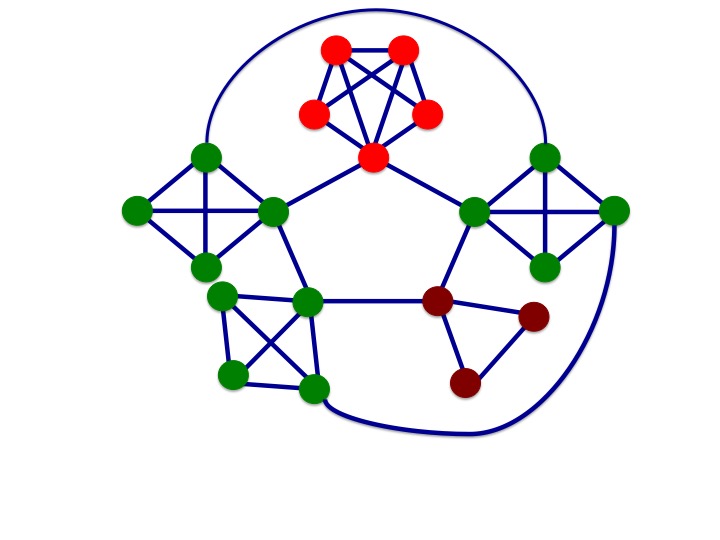}
   \end{minipage}
   \caption{\label{models} Simplified models of a network with (left) and without (right) an RCC. Red vertices have core number 4, green vertices have core number 3 and brown vertices have core number 2.  Note that the RCC is formed in the innermost core.}
\end{figure}

\section{Algorithm for Forming RCC}\label{sec:part3}
We now present a simple yet effective modification algorithm for inserting RCC into a network and conversely removing RCC from a  network containing it.


\noindent{\bf Rationale for the algorithm}: To explain the rationale for our algorithm, we present two simplified models of a network with a RCC and without a RCC (Figure~\ref{models}). If the network contains RCC, then the inner shells are expander like, and communities meet at the RCC. An example model conforming to this structure would be a large clique in the center surrounded by smaller cliques. 

In a network without a RCC, the majority of the communities do not meet through the inner core. This indicates that the inner core is not at any special position with respect to the paths connecting the communities. One example model of such a network would be a ring of cliques of different sizes. The smaller cliques can have connections between them. Here the highest core is at the side of the network rather than the center\footnote{We emphasize that these models are only idealized representations of the two types of networks, and more complicated connections occur for real-world networks. Nevertheless the principal idea is maintained, i.e., for networks with a RCC, the innermost core is at the center of the network.}. 
%
%

As per the example figure, to introduce a RCC, we can simply connect the high degree vertices across communities. The high degree of the vertices ensures that the clique (or near clique) formed by them will have higher core numbers. Joining communities ensures that the communities connect within this subgraph. In the network without an RCC in Figure~\ref{models}, we would connect all the vertices in the ring. 

Conversely, to destroy the RCC property of the network, we will simply delete the edges in the inner core, such that the connections between the communities are destroyed, and the highest numbered core moves away from the center.

\noindent{\bf Algorithm for forming RCC}: Our proposed approach for connecting the communities via  high degree vertices is however very expensive. This is because finding communities itself is a computationally intensive operation. A faster alternative is to simply connect (or disconnect) connections between the high degree vertices. This method works, because in networks without a RCC, high degree vertices within the same community are likely to be already connected. Therefore, any vertex pair connected as part of the modification algorithm will be in different communities. 

 Moreover, increasing the connections among the high degree nodes also brings all those (usually low degree) nodes that are neighbors of these high degree nodes closer in the network. Thus, the nodes in the innermost cores will have high centrality, as is a characteristic of networks with a RCC. On the converse side, for a network with a RCC, the high degree vertices will be in the inner core, so removing edges between them disconnects the cores.

It might seem from our approach that rich clubs of high degree vertices are also the rich centrality clubs. Figure~\ref{models} shows a counter example. The right hand graph has a rich club of high degree vertices but not a dense subgraph of high centrality vertices.

\noindent{\bf Experiments.} The pseudocode of the algorithm is given in Algorithm~\ref{algo:expansion}. Figure~\ref{fig:model}, plots the eigengaps of the networks before (blue lines) and after (green lines) the modification. To clearly compare between the original and modified networks, we plot the eigengaps for each shell, rather than over an aggregate of shells as done in Figure~\ref{fig:shell_fig}(c).
Note that for the networks that demonstrate the presence of RCC (AS, Bible and Software), the eigengap of the modified network is smaller, i.e., the green line is lower and has a steeper slope than the original network. For the networks that do not demonstrate the presence of RCC (Power, Protein and Facebook), the green line is higher, showing that the value of the eigengap increased and has a more gradual slope. We report the statistics of the modified network in Table~\ref{tab:modified}. The table clearly shows that our model also preserves the crucial structural properties of the original network, for e.g., the scale-free exponent $\alpha$ and the average degree.\footnote{We set the model parameter $h=30$. The results are similar for $h=20$ and $h=15$. $\gamma$ is set to 0.2} 




\begin{table}\small
\begin{tabular}{|c|c|c|c|c|c|c|c|} 
 \hline
 N/W & |V| & |E| & $ \alpha$ & $\mu(d_v)$ & $\mu(C_lC)$ & $\mu(BC)$ & \textit{LCN}\\  
 \hline
 As & 6474 & 12439 & 1.245 & 4.07 & 0.26 & 0.0012 & 9   \\ 
 \hline
 \hline
 Bible & 1773 & 8600 & 1.557 & 10.07 & 0.298 & 0.006  & 11   \\ 
 \hline
 Software & 1003 & 4400 & 1.236 & 8.85 & 0.339 & 0.004 & 9   \\ 
 \hline
Protein & 1870 & 2052 & 1.756 & 2.89 & 0.17 & 0.016 & 7  \\ 
 \hline
 Facebook & 7178 & 10349 & 2.311 & 2.82& 0.125 & 0.0123 & 6   \\ 
 \hline
 \hline
 Power & 4941 & 6698 & 2.344 & 2.71 & 0.0715 & 0.017 & 7    \\ 
 \hline
 
 \end{tabular}
 \caption{\label{tab:modified} Network statistics for the modified graphs. Note that the parameter values are comparable to that of the original networks in Table~\ref{dataset}.}
\end{table}

\vspace{-4mm}
\begin{algorithm}
\SetKwInOut{Parameter}{Parameters}
\KwIn{$G (V,E)$}
\KwOut {$E$}
\Parameter {$h,\gamma, flag$}
\CommentSty{\color{blue}}
Sort vertices in G based on decreasing degree\;
Select the top $h$ nodes based on degree\;
\If {$flag==1$}  {find $E_1$ possible edges that could be formed among the $h$ nodes}
\Else {find $E_1$ actual edges that are present among the $h$ nodes} 
Select $E_f$ edges randomly from $E_1$ where $|E_f| = \gamma * |E_1|$\;  
\lIf{flag == 1}{ $E \leftarrow E \cup E_f $ }
\lElse{$E \leftarrow E - E_f $}
\Return{$E$}\; 
\caption{\label{algo:expansion} Algorithm for increasing $(flag==1)$/decreasing ($flag==0$) expansion property.}
\end{algorithm}

\begin{figure*}[t!]
\centering
  \begin{tabular}{c c c}
   \includegraphics[width=0.3\textwidth,height = 2.8cm]{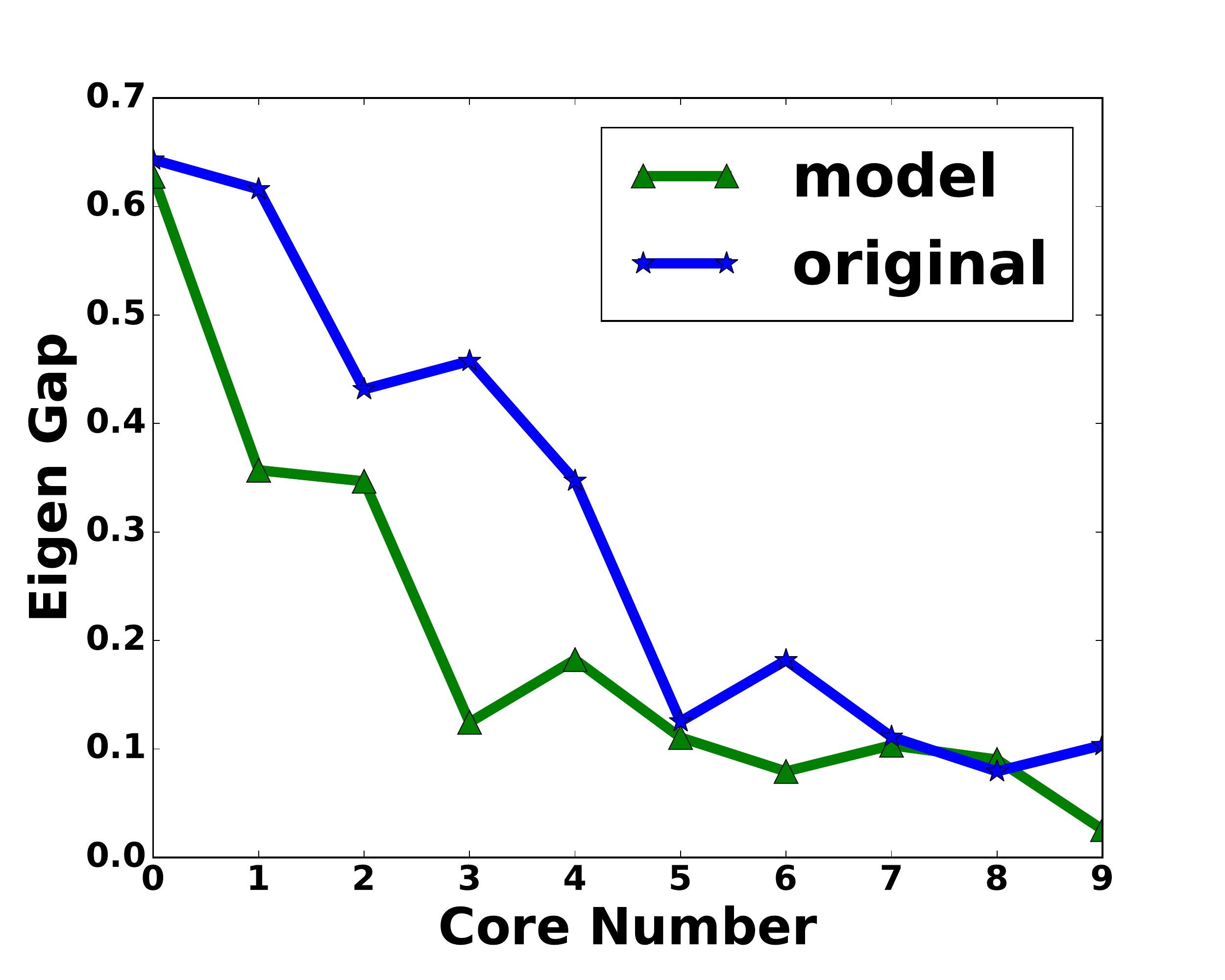}  &  \includegraphics[width=0.3\textwidth,height = 2.8cm]{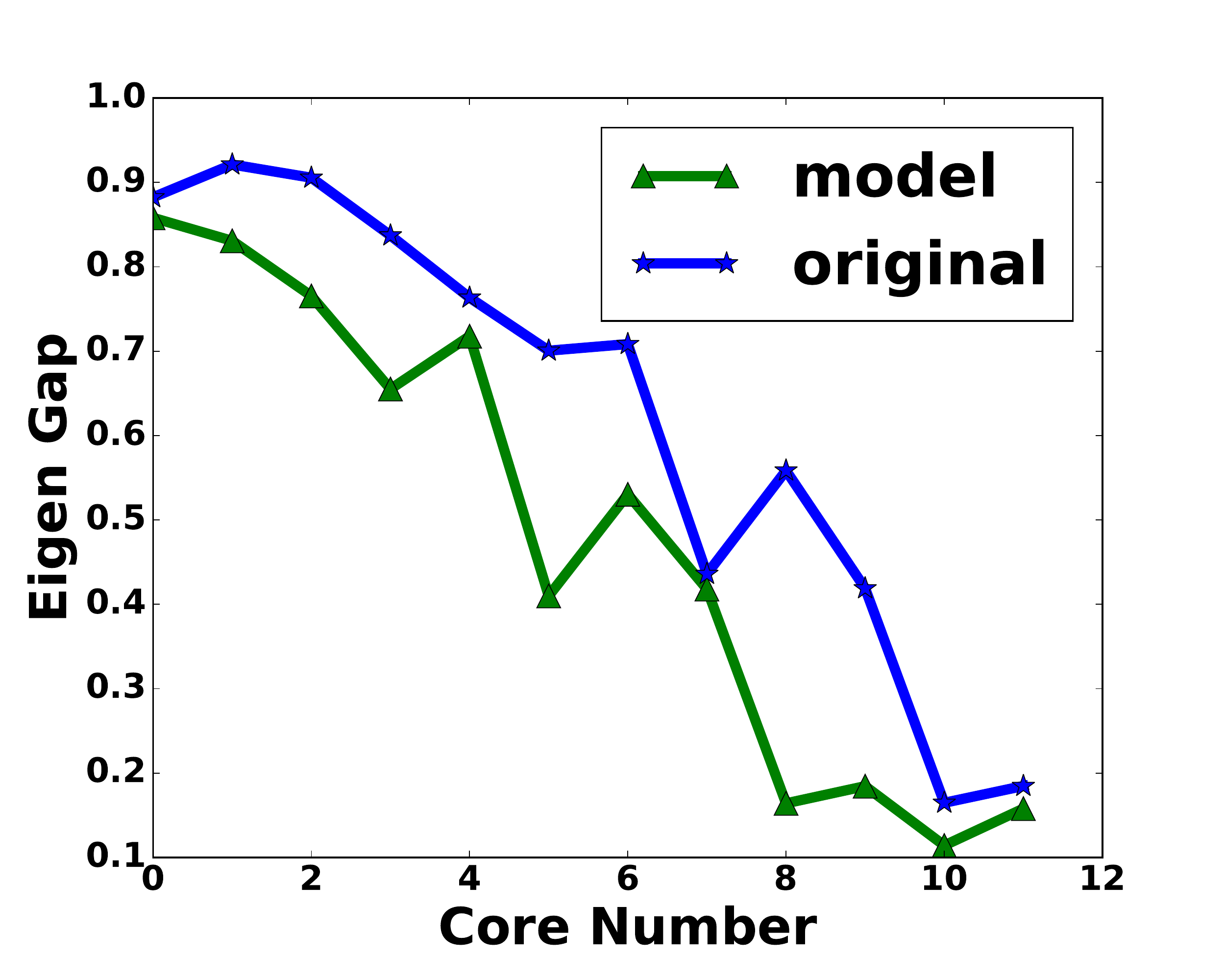} & \includegraphics[width=0.3\textwidth,height = 2.8cm]{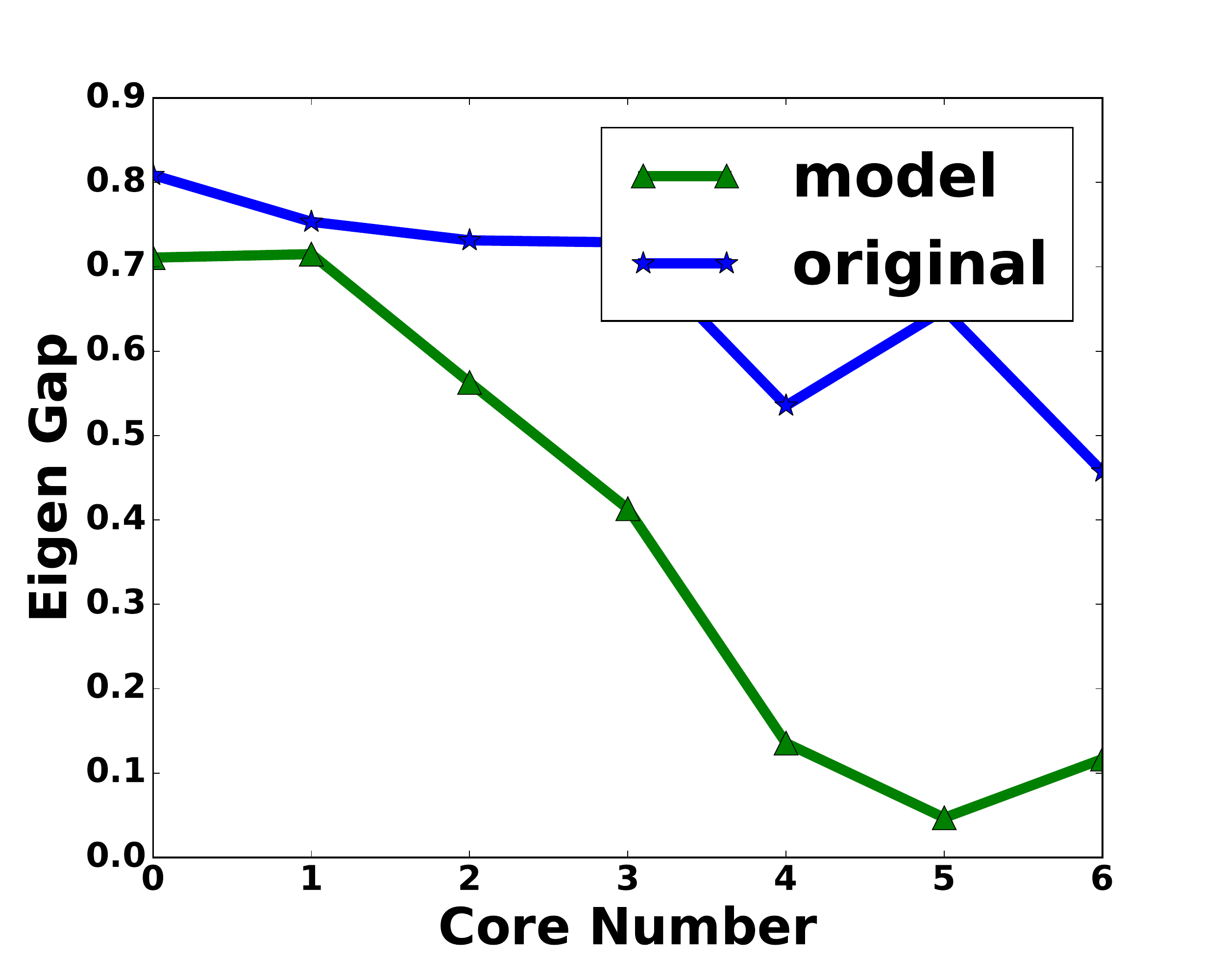} \\ 
      \includegraphics[width=0.3\textwidth,height = 2.8cm]{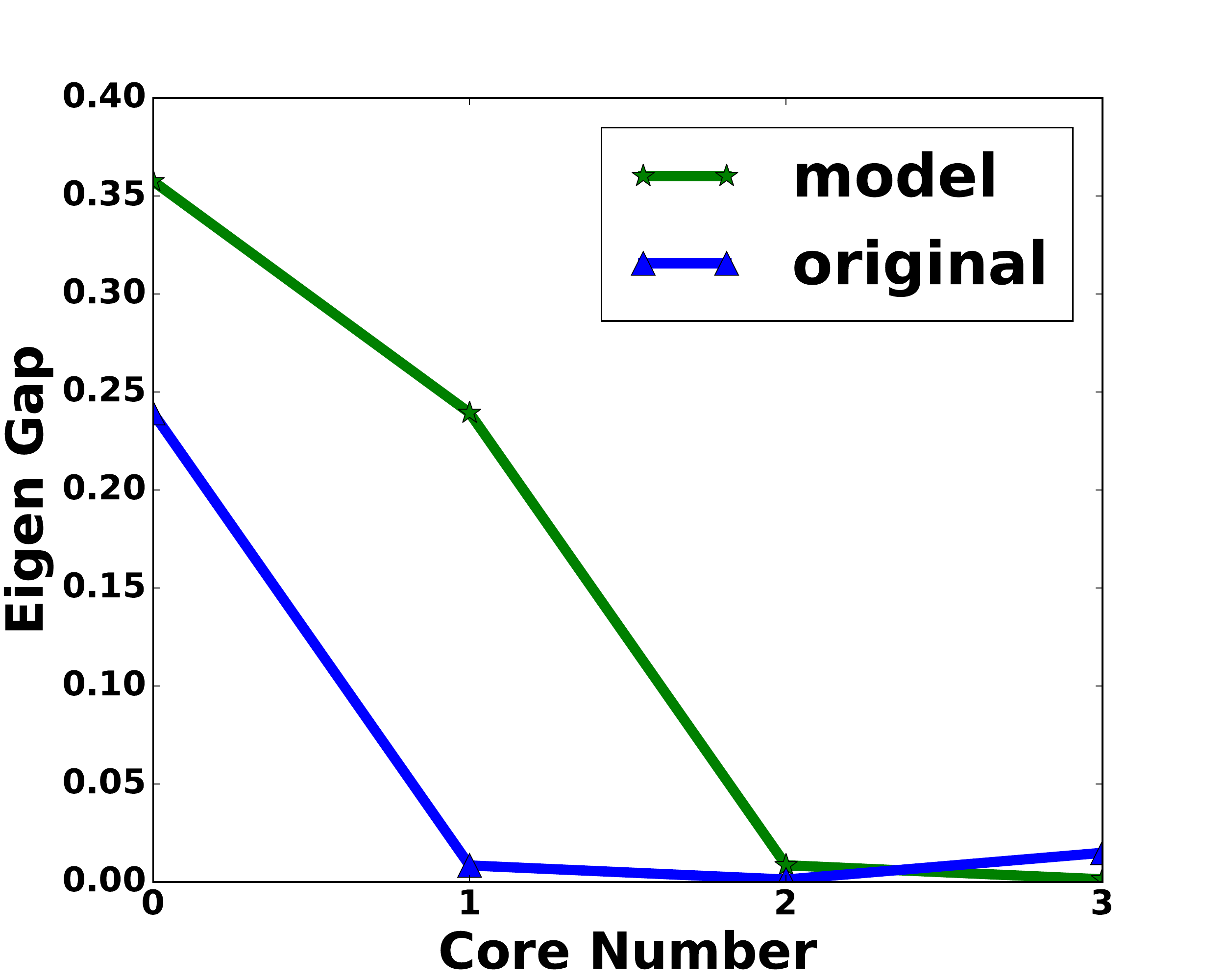} &  \includegraphics[width=0.3\textwidth,height =  2.8cm]{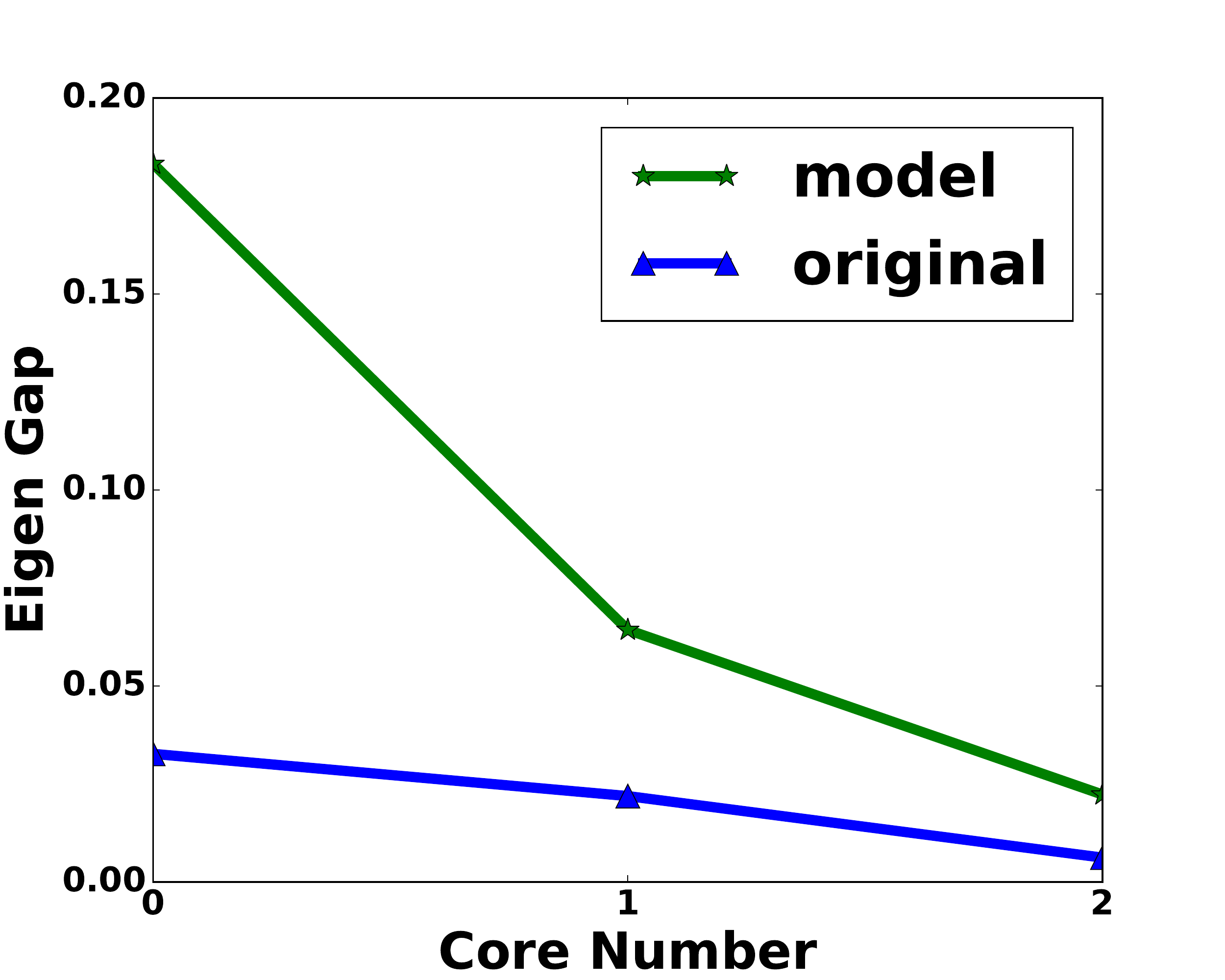} &  \includegraphics[width=0.3\textwidth,height = 2.8cm]{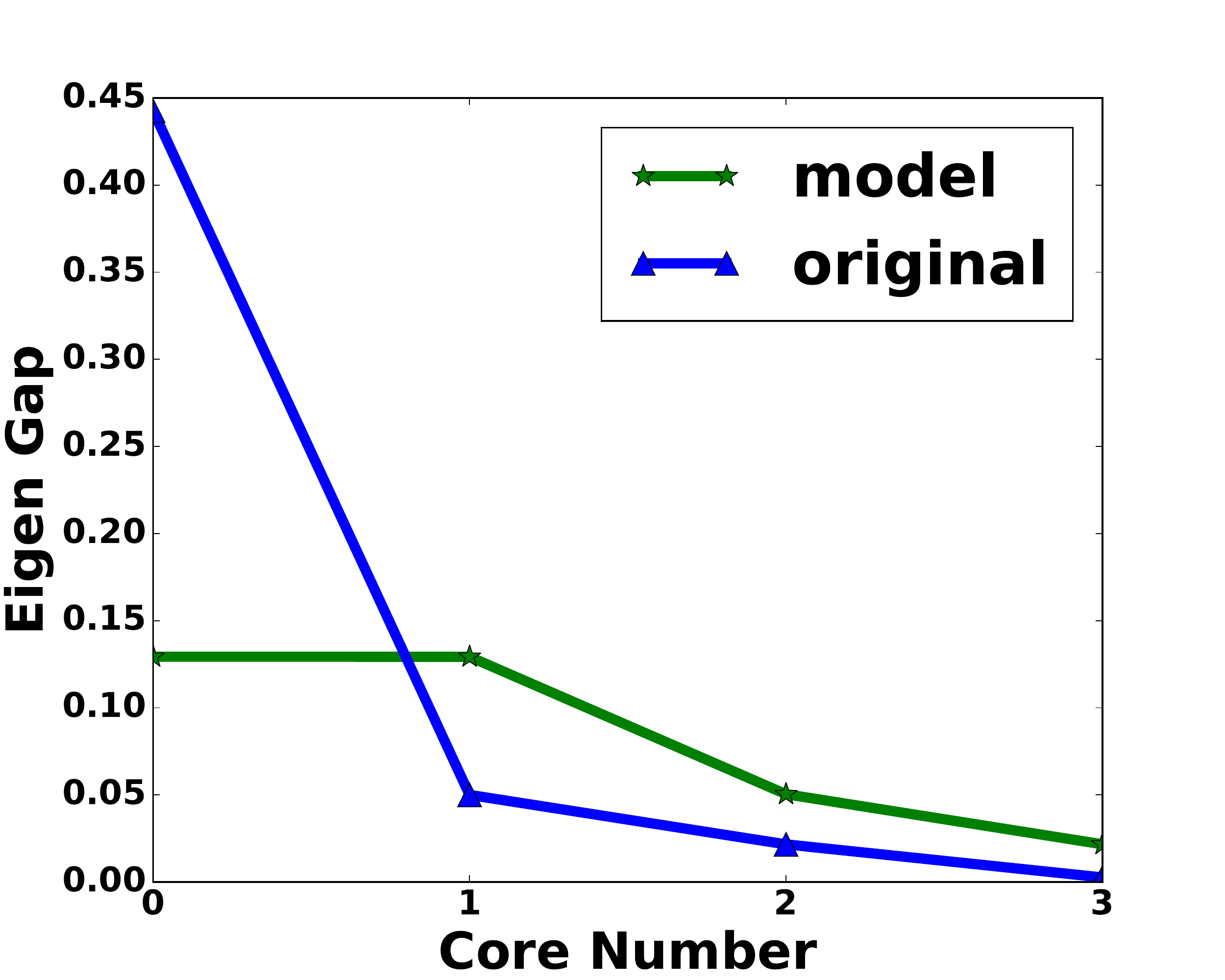} 
  \end{tabular}
  \caption{\label{fig:model} The outcome of the modification model. The first three networks (top panel) that originally demonstrate presence of RCC, i.e., AS, Bible and Software get transformed to networks with no RCC. The last three networks (bottom panel) that originally do not demonstrate the presence of RCC, i.e., Power, Protein and Facebook get converted to networks with RCC. These plots are similar to the eigen gap chart of Figure~\ref{fig:shell_fig}(c), except we show the eigengap over all the shells rather than in groups. The \textcolor{blue}{blue} (\textcolor{green}{green}) plot shows the eigengap  for the original (modified) network.}
  \label{correlation}
\end{figure*}

\section{Theoretical insights of definition}\label{sec:theory}
We now theoretically demonstrate how the three properties described in section~\ref{sec:part2} lead to the formation of rich centrality clubs. We consider an ideal network, where the vertices of each shell form a connected component and as per property 2,  the values of $\alpha$ are large enough such that each shell is an expander graph. Since an expander graph has no bottleneck, or clear partition, random graphs fulfill these criteria. We therefore assume that each subgraph, $S_{k}$ induced by a shell $k$ and its neighbors, is an Erdos-Reyni random graph, with $n_k$ vertices and $d_k$ average degree, and the probability of connection among a pair of vertices $p_k$; thus, $d_k \approx {n_k}{p_k}$.

In~\cite{chung2002average}, the authors prove several bounds on the average path length in a random graph $G(n,p)$. In particular, they state that if $np \ge c > 1$, then the average distance is  constant times $\frac{\log{n}}{\log{np}}$. Using this result, we assume that the average distance between two vertices in subgraph $S_k$ is $\frac{\log{n_k}}{\log{n_kp_k}} \approx \frac{\log{n_k}}{\log{d_k}}$.

Now consider a path between two vertices $v_1$ and $v_n$, with the sequence of vertices $v_1, v_2, \ldots,  v_x, \ldots, v_{n-1}, v_n$. Let the core numbers of these vertices be $l_1, l_2, \ldots l_x, \ldots l_{n-1}, l_n$. Let $l_x$ be the highest numbered shell  in this sequence. We assume that for all shortest paths between $v_1$ and $v_n$, $l_1 \le l_2 \le \ldots l_x \ge \ldots l_{n-1} \ge l_n$. 

This means that the shortest paths travel monotonically from a source low shell to the highest shell required, and then 
back from the high shell to the destination low shell. Note that this assumption allows the path to remain in the same shell throughout as well. However, paths that zig-zag from a high shell to a low shell and back to a high shell are not allowed. This rationale is based on the fact that since lower valued shells are sparser, it is more likely that the paths will connect through higher shells than lower shells.

With these assumptions in place we can state the following
\begin{lemma}
If there exists a shell $k_{x}$, such that, $\frac{\log{n_{k_y}}}{\log{d_{k_y}}} > r_{yz}+ \frac{\log{n_{k_z}}}{\log{d_{k_z}}} $, for all $k_y < k_x$ and for at least one $k_z \ge k_x$, then high centrality vertices are located in core $C_{x}$. 
\end{lemma}

\begin{proof}
Let $k_x$ be the lowest numbered shell that satisfies the equation in the lemma. Consider the path between two vertices $v$ and $u$. If either $v$ or $u$ is in $C_x$, then the path between them has to pass through $C_x$. If neither of the vertices are in $C_x$, we have to consider two cases.

First case, the two vertices in the same shell $k_a$, $x > a$, then on average, the distance between them will be $\frac{\log{n_{k_a}}}{\log{d_{k_a}}}$.  

Second case, the vertices from two different shells $k_a$ and $k_b$, $x > a > b$. On average the length of the shortest path will be $r_{ba}+ \frac{\log{n_{k_a}}}{\log{d_{k_a}}}$. This value is greater than the path simply going through $k_a$. 

Thus for both cases, if   $\frac{\log{n_{k_a}}}{\log{d_{k_a}}} > r_{ax}+ \frac{\log{n_{k_x}}}{\log{d_{k_x}}} $, the shortest path between any two vertices in the graph is on average going to pass though $C_x$. Thus the core $C_x$ will contain high closeness and betweeness centrality vertices. 

\end{proof}

As per property 1, $n_y > n_x$ and $d_y < d_x$, where $k_y <k_x$. Therefore the condition $\frac{\log{n_{k_y}}}{\log{d_{k_y}}} > \frac{\log{n_{k_x}}}{\log{d_{k_x}}} $ will hold for any two shells. To maintain the condition of the lemma, we have to ensure that the  distance from the steps to go from one shell to another, $r_{yz}$, is small enough. In other words, the steps to go from shell $y$ to core $z$ is smaller than the difference of their average distance. We have observed that for networks with RCC, $r_{yz}$, where $z$ is the innermost or second innermost core, the value is between 2-4. The number of nodes in the outer shells can go upto thousands, thus easily satisfying the condition.

It might seem that because finding the eigen value is an expensive operation, identifying networks with RCC would also be more expensive than simply finding the high centrality vertices. However, note that our lemma is based on the average path per shell. This metric can be computed in parallel for each shell, and is faster than computing the centralities over the whole network.


\section{Related Work}\label{sec:related}
In this paper, we bring together several concepts  from rich clubs, to core-periphery structure, to its application in information spreading and community detection, to expander graphs. To the best of our knowledge, this is the first paper to combine these different concepts within a single framework.


\noindent{\bf Rich club}:
Rich club structure is well studied in the context of infrastructure networks \cite{zhou2004rich,zhou2004accurately}. Rich clubs have also been shown to emerge in biological networks as well, e.g., brain networks, metabololic networks \cite{cinelli2018rich,bertolero2017diverse}. One of the recurring themes of research on characterizing rich club structure has been focused on distinguishing degree associativity and the rich club structure. 

\noindent{\bf Detection of cores}: Algorithmic detection of core periphery introduced by \cite{seidman1983network} is one of the most promising new area in network analysis. Several works such as \cite{holme2005core,rombach2014core} have presented that automatic techniques of separating the core nodes from the sparse periphery. 
Batagelj and Zaversnik \cite{batagelj2003m} proposed a $k$-core decomposition algorithm that requires $O(max(|E|, |V|))$ runtime and $O(|E|)$ space. In \cite{cheng2011efficient,khaouid2015k} the authors proposed modifications of the previous linear approach  which scales the computation to millions of nodes and billions of edges. In many networked systems, we are only concerned with estimating the importance of a subset of nodes instead of the entire network. In \cite{obrien2014locally}, the authors proposed a computation technique for computing the coreness for a node $u$ which only takes into account its $\delta$ neighborhood. Model based approaches have been presented in \cite{rombach2014core,zhang2015identification} where the authors designed objective functions to estimate the coreness of nodes.

\noindent{\bf  Correlation of coreness with centrality measures}: Coreness has been shown to be correlated with several centrality metrics. A strong Spearman's rank  correlation between degree and coreness has been presented in \cite{shin2016corescope,Shin2017} and the authors developed an anomaly detection system based on this correlation. In contrast, in\cite{li2015correlation} showed that core number has low Pearson's correlation with centrality metrics such as degree, closeness and betweenness. An explanation could be that while many nodes in the network could potentially have the same core number, they would typically tend to be different in terms of the centrality measures and thus Pearson's correlation would be low. A more accurate comparison would be to consider the overlap of the top ranked nodes based on core numbers and other centrality metrics, which is what we do here.

\noindent {\bf Coreness for community detection}: $k$-core decomposition outputs an ordered partition of the graph after processing it hierarchically. In \cite{giatsidis2014corecluster}, the authors proposed that this hierarchical information can be utilized by any graph clustering algorithm to obtain more meaningful partitions. In \cite{peng2014accelerating}, the authors proposed a framework to accelerate label computation for nodes by modularity maximization utilizing the $k$-core information. They estimate a maximum speedup of $80\%$ through rigorous experiments. 

\noindent{\bf Coreness for spreading}: The	core number though derived from degree, is a better indicator of the capacity for information dissemination. Strategically placed nodes, as detected by the $k$-core decomposition are able to spread information to a larger portion of the graph. This result has been shown by several works such as \citep{kitsak2010identification,bae2014identifying,pei2013spreading}. In \cite{cohen2008trusses,wang2012truss,rossi2015spread}, the authors apply $k$-truss decomposition which is a triangle based extension of $k$-core decomposition, with the objective of finding a refined set of influential nodes from all potential high core nodes. The authors show that $k$-truss decomposition extracts influential spreaders which can infect a large portion of target nodes within first few steps of the SIR epidemic model.

\noindent{\bf Expander graphs:} Analyzing graphs using spectral techniques has a long history \cite{maiya2014expansion,estrada2006spectral}. The main idea behind these approaches is to consider information about the spectrum of a matrix representation of the graph (mainly, the adjacency matrix or the Laplacian). It has been presented by several researchers \cite{malliaros2015estimating,krivelevich2017finding,kannan2004clusterings} that expansion properties of the graph provide crucial signals in understanding the degree of cohesion in the underlying subgraph structure. High expansion property results in being simultaneously sparse and tightly connected. A graph with such non-trivial structural property are called expander graphs and a comprehensive review on this topic can be found in~\citet{hoory2006expander}

\section{Conclusion}\label{sec:conc}
We study the properties of networks that demonstrate the presence of rich club of shortest path based high centrality nodes. We find that in these network, the vertices of the innermost core constitute the rich centrality club. Our main observations are as follows.

\begin{itemize}
\item The rich centrality clubs, if formed, are located in the inner cores of the network. These nodes can also be used as seed nodes for community detection. 
\item The networks with RCC typically have cores with expander-like structures. The density of the cores increase from inner to outer. The average number of hops to travel from an outer core to an inner core is small (2-4 hops).

\item  The nodes in the RCC of a network constitute very effective seeds for information diffusion. Further, presence of an RCC makes a network very resilient to small random structural perturbations. 
\item A simple model can convert a network with an RCC to a one without an RCC and vice versa. The model has just two global parameters to be tuned and builds on the idea of increasing the density/sparsity of connections among high degree, i.e., the innermost core nodes. 
\end{itemize}

Our experiments provide for the first time a deeper understanding about the rich club of high centrality vertices and their interplay with the core-periphery structure. In future, we aim to study dynamic networks to observe the evolution of RCCs over time. 
\bibliographystyle{ACM-Reference-Format}
\bibliography{ref} 

\end{document}